\documentclass[a4paper,12pt]{article}
\usepackage{iip-arxiv} 
 \pdfoutput=1 
\usepackage[utf8]{inputenc}
 \usepackage{graphics}
 \usepackage{braket}
 \usepackage{relsize}
 \usepackage{amsfonts,amssymb,amsmath,amsthm}
 \usepackage{graphicx,epsf,subfigure}
 \usepackage[dvipsnames]{xcolor}
\usepackage{bm}
\usepackage{setspace}
\graphicspath{{./figs/}}

\def\ds{\displaystyle}
\def\bea{\begin{array}{c}}
\def\ea{\end{array}}
\def\be{\begin{equation}\bea\ds}
\def\ee{\ea\end{equation}}
\def\bee{\begin{equation}\begin{array}{rcl}\ds}
\def\eee{\end{array}\end{equation}}

\DeclareMathOperator{\Tr}{Tr}

\newcommand{\q}[1]{\lq\lq#1\rq\rq}
\newcommand{\eq}[1]{(\ref{#1})}  
\newcommand{\fig}[1]{{\ref{fig:#1}}}  
\newcommand{\tab}[1]{\textbf{\ref{tab:#1}}}  
\newcommand{\Sec}[1]{\textbf{\ref{sec:#1}}}  

\newcommand{\diff}{\mathop{}\!\mathrm{d}}
\newcommand{\ii}{\textrm{i}\,}
\newcommand{\Mc}{\mathcal M}
\newcommand{\Cc}{\mathcal C}
\newcommand{\Rc}{\mathcal R}
\newcommand{\Pc}{\mathcal P}
\newcommand{\Lc}{\mathcal L}
\newcommand{\Dc}{\mathcal D}
\newcommand{\Hc}{\mathcal H}
\newcommand{\Sc}{\mathcal S}
\newcommand{\Tc}{\mathcal T}
\newcommand{\Kc}{\mathcal K}
\newcommand{\Uc}{\mathcal U}

\newcommand{\SL}{SL(2,\mathbb{Z})}

\newtheorem{proposition}{Proposition}

\makeatletter

\title{\bf Circuit Complexity of Knot States in Chern-Simons theory}
\author[1,a]{Giancarlo Camilo,}
\author[2,a,b]{Dmitry Melnikov,}
\author[3,a]{Fábio Novaes}
\author[4,a]{and Andrea Prudenziati}
\affiliation[a]{International Institute of Physics, Universidade Federal do Rio Grande do Norte\break Campus Universit\'ario, Lagoa Nova, Natal-RN 59078-970, Brazil}
\affiliation[b]{Institute for Theoretical and Experimental Physics, B.~Cheremushkinskaya 25, Moscow 117218, Russia}

\email{gcamilo at iip.ufrn.br}
\email{dmitry at iip.ufrn.br}
\email{fnovaes at iip.ufrn.br}
\email{andrea.prudenziati at gmail.com}

\abstract{
We compute an upper bound on the circuit complexity of quantum states in $3d$ Chern-Simons theory corresponding to certain classes of knots. Specifically, we deal with states in the torus Hilbert space of Chern-Simons that are the knot complements on the $3$-sphere of arbitrary torus knots. These can be constructed from the unknot state by using the Hilbert space representation of the $S$ and $T$ modular transformations of the torus as fundamental gates. The upper bound is saturated in the semiclassical limit of Chern-Simons theory. The results are then generalized for a family of multi-component links that are obtained by \lq\lq Hopf-linking\rq\rq~different torus knots. We also use the braid word presentation of knots to discuss states on the punctured sphere Hilbert space associated with 2-bridge knots and links. The calculations present interesting number theoretic features related with continued fraction representations of rational numbers. In particular, we show that the minimization procedure defining the complexity naturally leads to regular continued fractions, allowing a geometric interpretation of the results in the Farey tesselation of the upper-half plane. Finally, we relate our discussion to the  framework of path integral optimization by generalizing the original argument to non-trivial topologies.
}

\preprint{ITEP-TH-06/19}

\begin{document}

\maketitle


\section{Introduction}
\label{sec:intro}
\indent

The notion of computational complexity is original from computer science and generally refers to the minimum number of fundamental operations needed to implement a given task. In the context of quantum mechanics \cite{Aaronson:2016vto}, the typical task is that of using a unitary transformation $\Uc$ (a quantum circuit) to prepare a target state $|\Psi_\text{T}\rangle$ from a given reference state $|\Psi_\text{R}\rangle$, that is,
\be\label{task}
|\Psi_\text{T}\rangle = \mathcal{U} |\Psi_\text{R}\rangle\,.
\ee
The unitary $\mathcal{U}$ is assumed to be constructed as a sequence of fundamental operations 
called \emph{gates} that act on a small number of degrees of freedom (e.g., only $2$-qubit operations in a multi-qubit system). If we denote by $\mathcal{D}(\mathcal{U})$ -- the circuit depth -- the total number of gates used in a particular $\mathcal{U}$, the circuit complexity is simply
\be\label{cpxt}
\Cc(|\Psi_\text{T}\rangle,|\Psi_\text{R}\rangle)=\min_{\mathcal{U}}\mathcal{D}(\mathcal{U})=\mathcal{D}(\mathcal{U}_\text{optimal})\,,
\ee
where $\mathcal{U}_\text{optimal}$ is the circuit having the minimal number of gates.

In \cite{2006Sci...311.1133N}, Nielsen \emph{et. al.} introduced a nice geometric approach to quantum complexity. The basic idea was to assume that $\mathcal{U}$ is generated by some time-dependent Hamiltonian $H(t)$, such that each particular circuit can be understood as a particular path $\gamma$ in the space of unitaries,
\be
\label{Pexp}
\mathcal{U}_\gamma(t) \ = \mathcal{P}\exp\left(i\int_{\gamma} H(t) dt\right)\,.
\ee
Here the parameter $t$ can be taken to be in the range $[0,1]$ and, as dictated by \eq{task}, the paths are constrained to obey the boundary conditions $\mathcal{U}_\gamma(0)=\mathbb{I}$ and $\mathcal{U}_\gamma(1)=\mathcal{U}$. The main advantage is that, under appropriate definition of the circuit depth functional $\mathcal{D}[\mathcal{U}_\gamma(t)]$, the problem of finding the optimal circuit reduces to one of finding geodesics in a curved Riemannian manifold. The complexity \eq{cpxt} then is simply the length of this geodesic. To be precise, whenever it is possible to argue that the effective Hamiltonian $H(t)$ is of the form 
\be
H(t) = \sum_I Y_I(t)\mathcal{O}_I\,,
\ee
where $\{\mathcal{O}_I\}$ (the set of fundamental gates) are the generators of some Lie algebra, by using $\mathcal{D}[\mathcal{U}_\gamma]\equiv \int_{0}^1dt\sum_I|Y_I(t)|^2$ the problem reduces to the one of finding geodesics $\gamma$ (parametrized by the control functions $Y_I(t)$) in the corresponding Lie group manifold.

This geometric approach to circuit complexity has recently found its use in high energy physics motivated by two competing proposals for the complexity in conformal field theory states with a holographic dual, the so called \q{complexity $=$ volume} \cite{Stanford:2014jda} and \q{complexity $=$ action} \cite{Brown:2015bva} conjectures. The ground state complexity of a free scalar field was studied in \cite{Jefferson:2017sdb} (see also \cite{Chapman:2017rqy} for a related approach using the Fubini-Study metric), which was later generalized to coherent states \cite{Guo:2018kzl}, free fermions \cite{Hackl:2018ptj,Khan:2018rzm}, complex scalar \cite{Sinamuli:2019utz}, weakly interacting theories \cite{Bhattacharyya:2018bbv}, non-equilibrium states \cite{Alves:2018qfv,Camargo:2018eof,Ali:2018fcz,Liu:2019qyx}, thermofield double states \cite{Yang:2017nfn,Kim:2017qrq,Chapman:2018hou,Jiang:2018nzg}, and recently to ground states of lattice models displaying quantum phase transitions \cite{Ali:2018aon,Liu:2019aji} (see also \cite{Balasubramanian:2018hsu,Akal:2019ynl,Yang:2018nda,Yang:2018tpo,Yang:2018cgx} for further developments). A common feature of all these calculations is that they were carried out for free or weakly interacting theories. At the moment it remains unclear from a field theoretic point of view how to make sense of the complexity of states in strongly interacting theories, which would provide a better understanding of the holographic conjectures mentioned above. An exception here is the case of 2d CFTs, which has been studied recently in \cite{Caputa:2018kdj,Magan:2018nmu} and leads to results similar to the path integral optimization approach of \cite{Caputa:2017urj,Caputa:2017yrh,Bhattacharyya:2018wym,Takayanagi:2018pml}. 

In the present paper, we consider a two-step problem: first we define
what may be called \emph{topological complexity} for knots as the minimal number of modular $S$ and $T$ operations on the torus
$T^2$ that are necessary in order to produce a generic knot from a reference knot. Then
we consider the representation of these knots as states in the Hilbert
space of Chern-Simons theory with compact gauge group $G$ and level
$k$. We define in a similar way the circuit complexity of this knot
state as the size of the optimal circuit built from the unitary
representations of the modular transformations, $\Sc$ and
$\Tc$. The two in general will not be equivalent.

The quantum Chern-Simons theory is well-known to have only global topological degrees of freedom. Its quantization inside a solid torus gives rise to a finite-dimensional Hilbert space $\Hc(T^2;G,k)$ whose states correspond to 3-manifolds $\Mc$ having a boundary $\partial\Mc=T^2$. For $G=SU(N)$, the same Hilbert space is known to appear also in the quantization of the $SU(N)_k$ Wess-Zumino-Witten conformal field theory on $T^2$~\cite{Witten:1988hf,Bos:1989wa,Bos:1989kn,Labastida:1990bt}. The finite-dimensional nature of the Hilbert space makes the quantum theory remarkably simple (essentially an instance of quantum mechanics), despite of the intricate non-perturbative interactions appearing in the action. A similar story holds for a generic Riemann surface $\Sigma$ (eventually with punctures) at $\partial\Mc$.

A canonical basis on $\Hc(T^2;G,k)$ is constructed by inserting
circular Wilson loop operators colored with an integrable highest
weight representation of $G$ along the non-contractible cycle of the
solid torus \cite{Witten:1988hf}. The specific states we will consider
correspond to Wilson lines tied in the form of an arbitrary
\emph{torus knot} $\Kc_{n,m}$ and to a simple class of links made of
many torus knots \cite{Murasugi}. These torus knots are classified by
a pair of coprime integers $(n,m)$ that count how many times the knot
winds around the two fundamental cycles of the torus. Recalling that
different integer linear combinations of cycles are related by the
action of the modular group $P\SL$, one finds that the $(n,m)$ torus
knot can be obtained from a simple circular line parallel to the
non-contractible cycle (the \emph{unknot} $\Kc_{1,0}$) through a
modular transformation. Since $P\SL$ has a unitary representation on
$\Hc(T^2;G,k)$, this naturally defines a quantum circuit building the
state $|\Kc_{n,m}\rangle$ from the unknot state $|\Kc_{1,0}\rangle$.

The topological complexity corresponds to the shortest word of $P\SL$
generators yielding the desired knot transformation, which reduces the
problem of finding the complexity to a number theory problem. We also
briefly discuss the case of \emph{rational} (or \emph{2-bridge}) knots
and links \cite{Murasugi} by adopting the standard presentation of
knots as the closure of braid words. For the specific case of 2-bridge
knots, the problem can again be translated into the same one of $P\SL$
generators. 

When the study of complexity is extended to the knot states in the quantum Hilbert space the representation constraints need to be taken into account\footnote{We thank the anonymous referee for reminding us about it.}. These constraints lead to non-trivial linear relations between the states. In simple words, states with $(n,m)$ and $(q,p)$ can become
equivalent if they are related by one such constraint. Therefore the analysis based on a simple counting of the
word generators, in general, only gives an upper bound on the circuit
complexity. 

Keeping in mind this caveat we denote $U_{n,m}^\text{min}$ the optimal topological circuit made of the smallest number of $P\SL$ generators needed to construct a generic knot, and $\Uc_{n,m}^\text{min}$ the optimal quantum circuit operator in the unitary representation acting on the corresponding Hilbert space. We show that different realizations of a topological circuit $U_{n,m}$ are associated with different continued fraction decompositions of the rational $\tfrac{n}{m}$, showing an interesting interplay between our problem and number theory. In particular, we prove (see Proposition \ref{prop1}) that the optimal circuit $U^\text{min}_{n,m}=T^{a_1}S\cdots T^{a_r} S$ corresponds to $a_i=(-1)^{i+1}b_i$, where $b_i>0$ are the regular (or Euclidean) continued fraction coefficients. The \lq\lq topological complexity\rq\rq~of this knot (the minimal number of $S$ and $T$ transformations needed to produce it from the unknot) is then simply \begin{align}\label{Cnmintro}
    C_{n,m} = \sum_{i=1}^r(b_i+1) + |f|\,,
\end{align}
where $|f|$ is a further contribution due to possible framing of the knot. The corresponding quantum complexity $\Cc_{n,m}$ of the knot state $|\Kc_{n,m}\rangle$ will in general be lower than \eq{Cnmintro} due to additional constraints on the corresponding optimal quantum circuit $\Uc_{n,m}^{opt}$, i.e. $\Cc_{n,m}\le C_{n,m}$.

We propose a geometric interpretation of $C_{n,m}$ in terms of geodesic paths on a graph connecting rational numbers (the Farey graph), which has a natural representation in the upper-half plane. We also discuss a related interpretation in the Stern-Brocot tree of rational numbers, which is related to the view of topological complexity in terms of geodesics on the Cayley graph of $S$ and $T$ generators \cite{Lin:2018cbk}. We find evidence that, in addition to Euclidean continued fractions, the ancestral path continued fractions introduced in \cite{beardon2012} also yield an optimal circuit, indicating that the complexity appears as an invariant related to different continued fractions. 

The upper bound on quantum state complexity is saturated in the semiclassical limit of the Chern-Simons theory. Intuitively, this topological complexity of semiclassical knot states in Chern-Simons is related to the old problem of classifying knots, even though nowadays there are many known knot invariants designed to solve this problem.\footnote{As in the case of the degeneracy of torus knots states mentioned above, knot invariants can appear degenerate on inequivalent knots. However, by adjusting parameters $k$, $G$, or the representation $R$ coloring the knot, one seems to always be able to find a polynomial distinguishing a given knot from an arbitrary collection of other knots \cite{Murasugi}.} We find that an obvious extension of the torus knot discussion to rational links (understood as four-strand braid closures) recovers the same results for complexity when applied to torus knots, which are special members of the rational class.

We also discuss the relation with the path integral optimization approach to complexity developed in \cite{Caputa:2017urj,Caputa:2017yrh,Bhattacharyya:2018wym,Takayanagi:2018pml} by generalizing the construction to spaces whose topology is non-trivial. In this way, the optimization procedure is seen to include also a variation over the moduli space and a connection with our problem is then established.

We note that the relation of knots to quantum Chern-Simons theory as described above puts them in the quantum information context. Indeed, the connection between knots and quantum computing has long been appreciated -- they are alike quantum algorithms and their complexity should tell us about the complexity of the underlying quantum tasks \cite{10.2307/26061414,2008RvMP...80.1083N,Melnikov:2017bjb}. Complexity, in the sense discussed here, can be related to complexity (in an abstract sense) of certain algorithms for computing knot invariants, cf.~\cite{2009arXiv0908.0512K,Cherednik:2011nr,Freedman:2000rc}.

The remainder of this paper is organized as follows. In Section~\Sec{CS}, we review the construction of the Hilbert space of Chern-Simons theory on the torus and give a detailed definition of torus knots states to be studied in the sequence. We also discuss the subtleties involved when going from the modular group to its unitary representation and the eventual conditions on the parameters of the theory which allow the given set of torus knot states to be non-degenerate. In Section~\Sec{comp}, we compute the topological complexity of torus knots after proving a set of statements about continued fractions and then discuss the results through examples. We also give a geometric interpretation of the result and discuss relations with mathematical results on geodesic paths on the Farey graph. A comprehensive discussion then follows on the circuit complexity of knot states and its relation with the topological result, which constitutes an upper bound as anticipated above. In Section~\Sec{general} we extend the discussion to the case of other knots and links.  In Section \Sec{pathintegral} we discuss a connection with the path integral optimization approach to complexity. We summarize our findings in the closing Section~\Sec{conclusions}. Some additional details about torus knot states are left for Appendix~\ref{sec:appendix}.

\section{Knot States in Chern-Simons theory}
\label{sec:CS}

\subsection{Knot complement states}

\indent

The Chern-Simons theory with gauge group $G$ and level $k$, denoted by $G_k$, is defined on a (compact, connected, oriented) 3-manifold $\Mc$ by the action
\begin{align}\label{CSaction}
S_{CS}[A] = \frac{k}{4\pi} \int_{\Mc} \Tr\left(A\wedge \diff A+\frac{2}{3}A\wedge A\wedge A\right)\,,
\end{align}
where $A=A_\mu\diff x^\mu$ is the gauge field and the trace is taken in the fundamental representation of the Lie algebra of $G$. 
The level $k$ is an integer in order to ensure gauge invariance of the path integral defining the quantum theory \cite{Witten:1988hf}. 

This action is topological in the sense that it is independent of the metric chosen in $\Mc$. As a consequence, the expectation value of any gauge invariant and metric independent observable of the theory defines a topological invariant in $\Mc$. The natural example is the Wilson loop operator associated with an oriented closed curve (e.g., a \emph{knot}) $\Kc$,
\begin{align}
    W_\Rc(\Kc) = \Tr_\Rc \, \Pc\exp\left(\oint_\Kc A\right)\,,
\end{align}
obtained by tracing in a given representation $\Rc$ the holonomy of the gauge field around $\Kc$. More generally, the expectation value of any product $W_{\Rc}(\Lc)\equiv\prod_i W_{\Rc_i}(\Kc_i)$ of Wilson loops computes a topological invariant of the \emph{link} $\Lc=\coprod_i \Kc_i$ obtained by joining the (non-intersecting) knots $\Kc_i$.\footnote{The symbol $\coprod_{i}U_{i}$ means the disjoint union of the sets $U_{i}$.} This is calculated as usual by the path integral 
\begin{align}\label{W(L)}
    \big\langle W_{\Rc}(\Lc)\big\rangle_\Mc = \int_\Mc\Dc A\,\left(\prod_i W_{\Rc_i}(\Kc_i)\right)\,e^{\ii S_{CS}[A]}\,,
\end{align}
where the normalization factor ${Z(\Mc)}^{-1}$ is omitted for brevity $\big({Z(\Mc)}=\int_\Mc\Dc A\,\,e^{\ii S_{CS}[A]}\big)$. When the gauge group $G$ is $SU(2)$ and the $\Rc_i$ are all fundamental representations, this reduces to the celebrated Jones polynomial of $\Lc$ \cite{Witten:1988hf}. Similarly, for $SU(N)$ and $SO(N)$ one gets the HOMFLY-PT~\cite{Freyd:1985dx,Przytycki:1987} and the Kauffman~\cite{Kauffman:1990} polynomials, respectively. In general \eq{W(L)} gives access to infinitely many link invariants as the gauge group and representations are changed.

We are interested in Chern-Simons theory defined on a topological 3-manifold $\Mc$ with a 2-torus as a boundary, $\partial\Mc=T^2$.  Any such $\Mc$ can be understood as the \emph{knot complement} of some knot $\Kc$ in a closed 3-manifold, which here for simplicity we take to be the 3-sphere. Namely, we will be interested in 3-manifolds $\Mc=S^3\backslash \Kc_\text{tub}$ constructed by removing from $S^3$ a small tubular neighborhood $\Kc_\text{tub}$ of a knot $\Kc$. 
The simplest example is provided by the trivial knot or \emph{unknot}, in which case the region $\Kc_\text{tub}$ is a simple (i.e., unknotted) solid torus and its complement $\Mc$ turns out to be another solid torus \cite{HatcherBook}. For a non-trivial knot the situation is illustrated in Figure \fig{knotcomplement}. 

The Chern-Simons path integral on $\Mc$ defines a state in the Hilbert space $\Hc(T^2;G,k)$ associated with the $T^2$ boundary. 
We shall denote this by $|\Kc\rangle$ (because $|S^3\backslash \Kc_\text{tub}\rangle$ would be too cumbersome) and refer to it as the \emph{knot complement state} associated to $\Kc$.\footnote{The construction here is a particular case of the \emph{link states} used in \cite{Salton:2016qpp,Balasubramanian:2016sro,Dwivedi:2017rnj,Balasubramanian:2018por,Melnikov:2018zfn} to study the relation between entanglement and topology in Chern-Simons, where the knot $\Kc$ is replaced by a link $\Lc=\coprod_i \Kc_i$ and the resulting manifold $\Mc$ has multiple torus boundaries.} 
Since the theory is topological, states constructed in this way only depend on the topology of $\Mc$, not on its precise geometry. Therefore, different quantum states on the torus can be constructed by considering different knot complements.

\begin{figure}
    \centering
    \includegraphics[width=.4\textwidth]{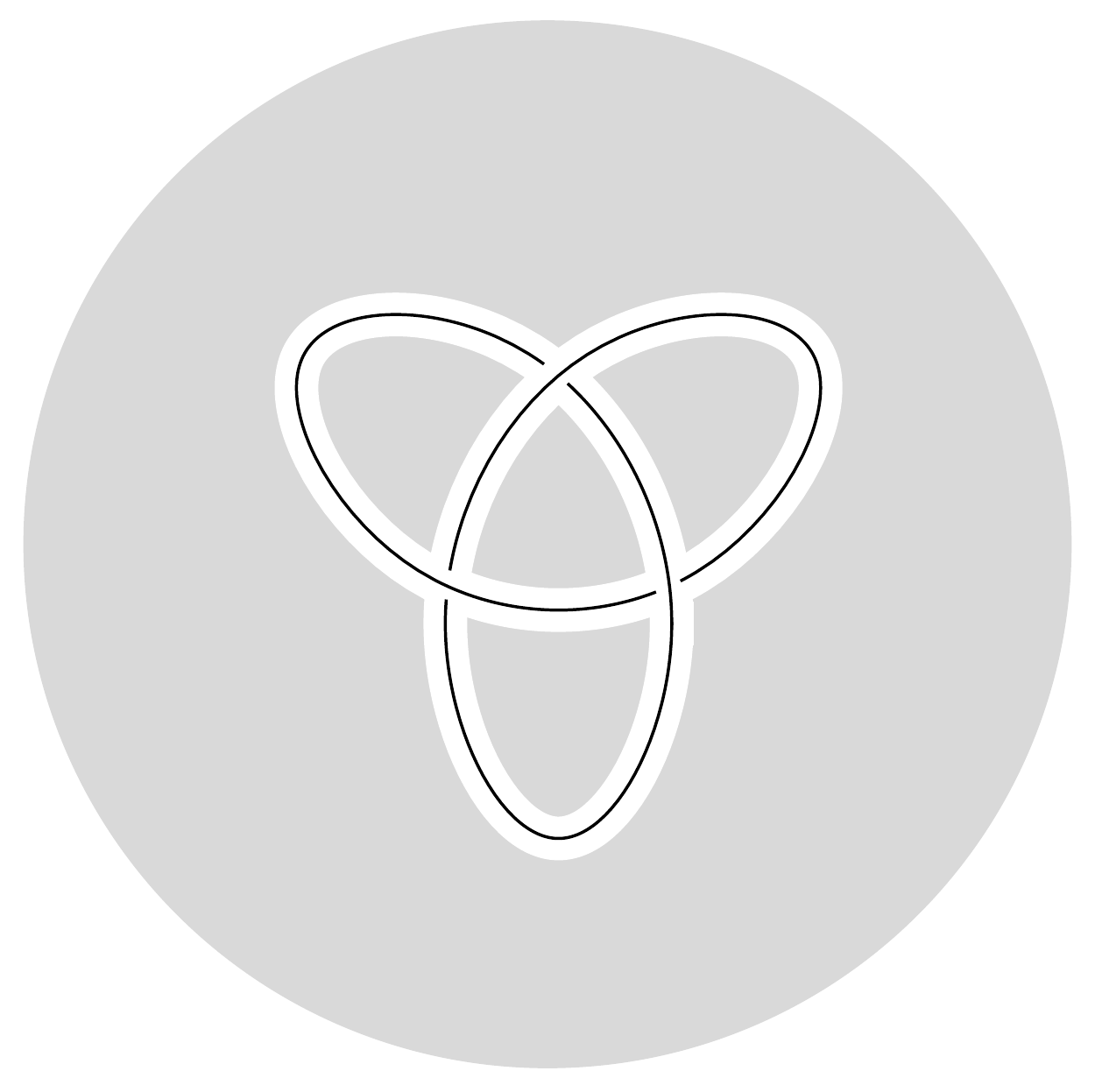}
    \caption{The knot complement manifold $\Mc$ corresponding to the trefoil knot (gray region). It is constructed by removing from $S^3$ a small tubular neighbourhood $\Kc_\text{tub}$ (white region) of the knot (black curve).}
    \label{fig:knotcomplement}
\end{figure}

A canonical basis for the torus Hilbert space $\Hc(T^2;G,k)$ can be constructed taking $\Mc$ to be a solid torus (the complement of the unknot, in the spirit above) and inserting Wilson lines in it \cite{Witten:1988hf}. Namely, the path integral on the solid torus with a Wilson line in the representation $\Rc_j$ inserted along the non-contractible cycle defines a state $|j\rangle$ on the boundary (its conjugate $\langle j|$ corresponds to the insertion of a Wilson line in the conjugate representation $\overline{\Rc}_j$ with inverted boundary orientation). 
An orthonormal basis for $\Hc(T^2;G,k)$ then consists of the set $\big\{|j\rangle\big\}$ where $j$ runs over integrable highest weight representations of the gauge group $G$ at level $k$. For instance, in the case of $U(1)_k$ these integrable representations are labelled by an integer $j=0,1,\ldots,k-1$, while for $SU(2)_k$ they are labelled by a half-integer $j=0,\tfrac{1}{2},\ldots,\tfrac{k}{2}$. The resulting Hilbert spaces $\Hc(T^2;U(1),k)$ and $\Hc(T^2;SU(2),k)$ have dimensions $k$ and $k+1$, respectively. 

The knot complement state $|\Kc\rangle\in\Hc(T^2;G,k)$ obtained by path-integrating over $\Mc$ can be expanded in the $|j\rangle$ basis as 
\begin{equation}\label{ketK}
    |\Kc\rangle = \sum_{j}\psi_j(\Kc)\,|j\rangle\,.
\end{equation}
The coefficients $\psi_j(\Kc) = \langle j|\Kc\rangle$ are computed by the inner product corresponding to gluing  the two manifolds that define $|\Kc\rangle$ and $\langle j|$ along their common $T^2$ boundary. The result is a sphere with the Wilson line inserted, which is nothing but the knot invariant \eq{W(L)} on $S^3$, 
\begin{equation}\label{psi(K)}
    \psi_j(\Kc)= \big\langle W_{\overline{\Rc}_j}(\Kc)\big\rangle_{S^3}\,.
\end{equation}
In other words, the state $|\Kc\rangle$ contains all the Wilson loop knot invariants of the knot $\Kc$ at level $k$. 

\subsection{The framing ambiguity}

Strictly speaking, the invariants \eq{W(L)} are only well-defined for \emph{framed knots} (links) \cite{Witten:1988hf,Guadagnini:1990uw}. Informally, a framed knot $\Kc$ is just the usual knot $\Kc$ constructed using a ribbon instead of a dimensionless string. In other words, it is the object obtained by stretching the curve $\Kc$ a little bit at each point along a direction specified by a normal vector field $\mathbf{v}$ (called the \emph{framing}), as illustrated in Figure \fig{framedknots}. The resulting framed knot has twists in the ribbon, the number $f$ of which is called \emph{framing number} or \emph{self-linking number} of $\Kc$.  The issue becomes clear in the case of $U(1)$ Chern-Simons theory, where the action \eq{CSaction} is quadratic and correlators can be computed in closed form. Namely, for a link $\Lc=\coprod_\alpha\,\Kc_\alpha$, \cite{Witten:1988hf}
\begin{align}\label{vevambiguous}
    \big\langle W(\Lc)\big\rangle_{S^3} = \exp\left(\frac{2\pi\ii}{k}\sum_{\alpha,\beta}n_\alpha n_{\beta}\ell_{\alpha\beta}\right)\,,
\end{align}
where the integer $n_\alpha$ labels the representation of the corresponding $\Kc_\alpha$ and
\be\label{linking}
\ell_{\alpha\beta} = \frac{1}{4\pi}\oint_{\Kc_\alpha}dx^\mu \oint_{\Kc_\beta}dy^\nu\,\epsilon_{\mu\nu\rho}\frac{(x-y)^\rho}{|x-y|^3}\,
\ee
is the Gauss linking number, a well-known topological invariant counting how many times the knots $\Kc_\alpha$ and $\Kc_\beta$ ($\alpha\ne\beta$) wind around each other. There is an inherent ambiguity in \eq{vevambiguous}, however, coming from the contributions when $\alpha=\beta$, where a prescription is needed to deal with the integration over coincident points. The problem is similar to the ambiguity in the definition of the composite operator $(\oint_{\Kc_\alpha} A)^2$ in the quantum theory. Even though a careful inspection of the integral in $\ell_{\alpha\alpha}$ shows that it is well-defined and finite \cite{Guadagnini:1989am}, it turns out to be metric-dependent and hence not invariant under deformations of $\Kc_\alpha$, which spoils the desired topological invariance of the result. A regularization prescription that restores this topological property amounts to introducing a framing for $\Kc_\alpha$, as explained above, and defining $\ell_{\alpha\alpha}$ as the linking number between $\Kc_\alpha$ and its framing, which is precisely the self-linking number $f_\alpha$. 

In general, there is no natural choice of framing and generic observables will depend on this choice (see, e.g., \cite{Marino:2001re} for the framing dependence of Wilson loops in $SU(N)$ Chern-Simons).\footnote{It is interesting to mention that the entanglement structure of link complement states has been shown to be framing-independent \cite{Balasubramanian:2016sro,Balasubramanian:2018por}.} However, this by no means takes away the merit of Chern-Simons theory, since the transformation rule for expectation values of Wilson loops under a change of framing is well-defined. Namely, if the knot $\Kc_\alpha$ has its framing shifted by $t$ units, $\ell_{\alpha\alpha}$ is increased by $t$ and, as a result, it is clear from \eq{vevambiguous} that the Wilson loop picks up a phase factor,
\begin{align}\label{framing}
    \big\langle W(\Lc)\big\rangle_{S^3} \,\longrightarrow\, \exp\left(2\pi\ii t\, h_\alpha\right)\big\langle W(\Lc)\big\rangle_{S^3}\,,\qquad h_\alpha\equiv n_\alpha^2/k\,.
\end{align}
Even though we have only discussed here $U(1)_k$, the same conclusion holds in $SU(N)_k$ as well, where $h_\alpha$ in that case is the conformal weight of the Wess-Zumino-Witten primary field corresponding to the representation $\Rc_\alpha$ \cite{Witten:1988hf}. 

\begin{figure}[ht]
\centering
    \subfigure[]{\includegraphics[width=.3\textwidth]{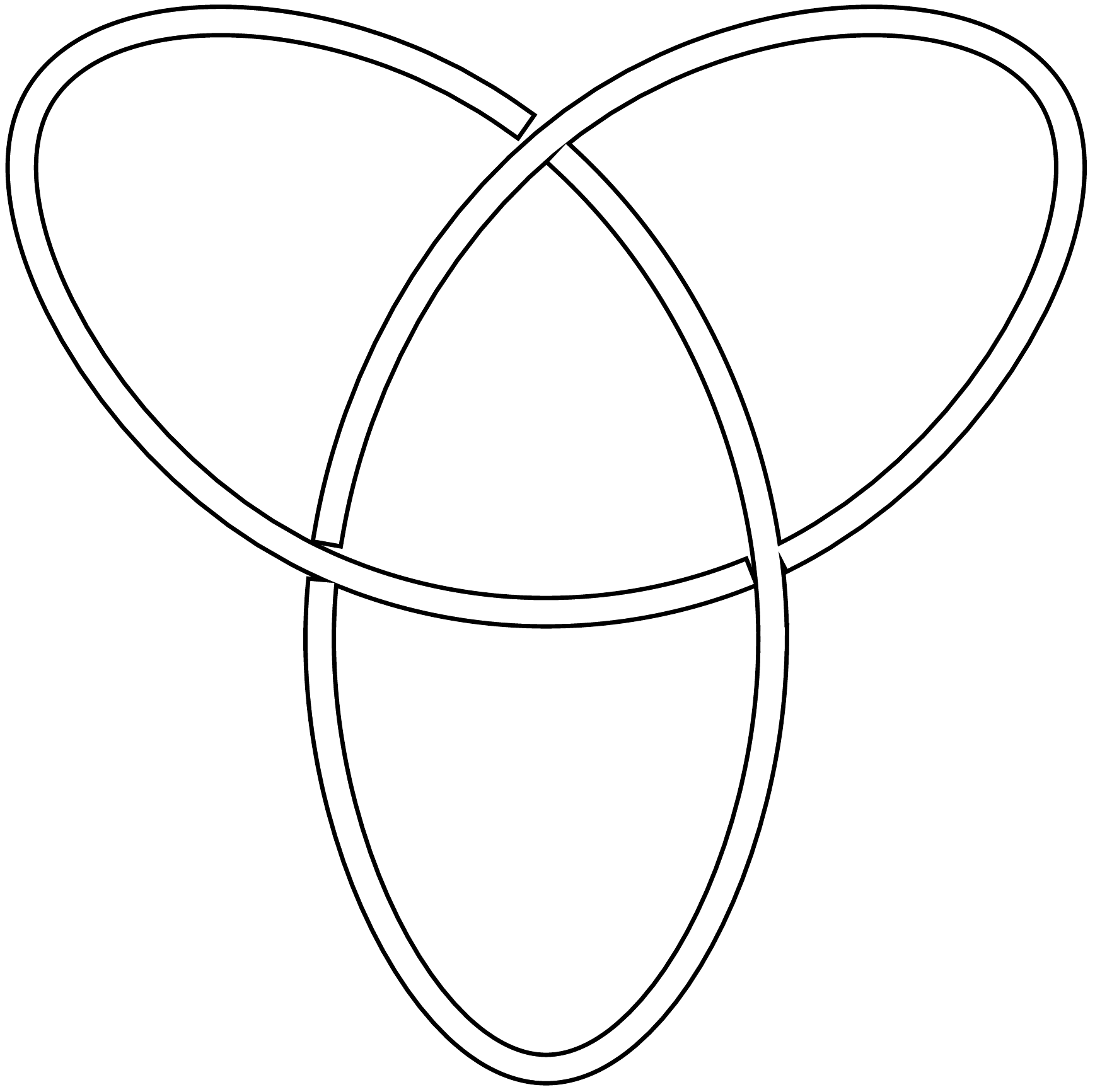}}\qquad\qquad
    \subfigure[]{\includegraphics[width=.3\textwidth]{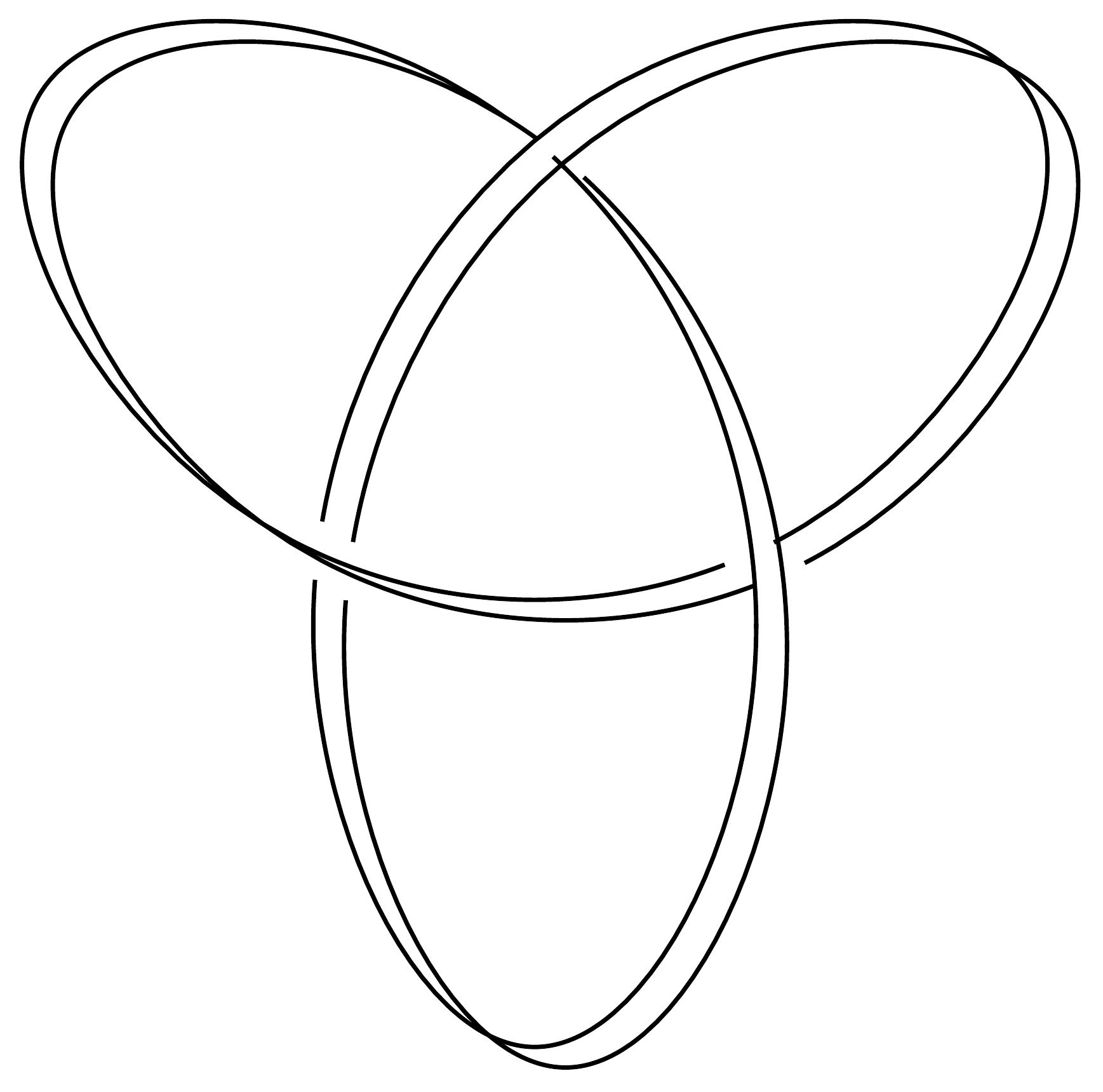}}
    \caption{Two different framings of the trefoil knot. (a) is the simplest to visualize, the so called \emph{blackboard framing}, where the normal vectors point all to the same vertical direction and the ribbon lies flat on the projection plane. According to Calugareanu's theorem \cite{Kaul:1999je}, the self-linking number in this case is the \emph{writhe number} of the knot ($w=3$ here). (b) shows another choice of framing having $4$ extra twists on the ribbon (in general, any framing can be drawn as a blackboard framing with a number of extra twists). Note that each twist may increase or decrease the self-linking number, depending on whether it is a positive or negative twist. In particular, the self-linking number can always be set to zero by an appropriate number of twists -- this defines the so called \emph{canonical framing}.}
    \label{fig:framedknots}
\end{figure}

\subsection{Modular transformations and torus knots}
\label{sec:modular}

There is a natural set of unitary transformations acting on the Hilbert space $\Hc(T^2;G,k)$. They correspond to the unitary representations of the modular group $P\SL$ of large diffeomorphisms of the torus. The modular group is generated by the $S$ and $T$ transformations
\begin{equation}
    S=\begin{pmatrix}
    0 &-1\\
    1 &0
    \end{pmatrix}\,,\qquad
    T=\begin{pmatrix}
    1 &1\\
    0 &1
    \end{pmatrix}
\end{equation}
that act on the torus modular parameter $\tau$ as $S:\tau\to-\frac{1}{\tau},\,T:\tau\to\tau+1$. They satisfy $S^2 = (ST)^3 = 1$ and the $P$ in front of $P\SL$ means that $\SL$ matrices $M$ and $-M$ should be identified. It is easy to see that $S$ exchanges the two fundamental cycles of the torus, while $T$ generates twists around the  contractible cycle (the so-called \emph{Dehn twists}). Any modular transformation can be written as a sequence of $S$ and $T$ transformations. 
These act naturally on the torus Hilbert space through their unitary matrix representations $\Sc$ and $\Tc$ (of dimension $
\dim\Hc(T^2;G,k)$), which take particularly simple forms in the $|j\rangle$ basis. For instance, in the $U(1)_k$ theory, they are given by \cite{Gepner1986}
\begin{equation}\label{STU(1)}
\Sc_{j_1,j_2} = \frac{1}{\sqrt{k}} e^{2\pi\ii\frac{j_1 j_2}{k}},\quad \Tc_{j_1,j_2} = e^{2\pi \ii h_{j_1} } \delta_{j_1,j_2}
\end{equation}
with $h_{j} = \frac{j^2}{2k}$, while for $SU(2)_k$ they read 
\begin{equation}\label{STSU(2)}
\Sc_{j_1,j_2} = \sqrt{\frac{2}{k+2}}\sin\left(\frac{\pi (2j_1+1)(2j_2+1)}{k+2}\right),\quad \Tc_{j_1,j_2} = e^{2\pi \ii h_{j_1}} \delta_{j_1,j_2}
\end{equation}
with $h_{j} = \frac{j(j+1)}{k+2}$. 

From a knot theory perspective, a sequence of $S$ and $T$ diffeomorphisms transforms a circular Wilson loop (the unknot) inside the solid torus into an arbitrary \emph{torus knot}. 
We recall that torus knots, which we denote by $\Kc_{n,m}$ or by the pair $(n,m)$, are knots that can be drawn on the surface of a torus without self-intersections. Non-trivial knots are labelled by two coprime\footnote{When they are not coprime, $\Kc_{n,m}$ denotes instead the \emph{torus link} made of $N=\text{gcd}(n,m)$ copies of $\Kc_{\frac{n}{N},\frac{m}{N}}$.}  numbers $(n,m)$ that count how many times the knot winds around the two fundamental cycles of the torus (the non-contractible and contractible one, respectively). With no loss of generality, they can be taken to be  $n>m>1$.\footnote{Namely, it is not hard to check that $i)$ $\Kc_{n,m}\equiv\Kc_{m,n}$; $ii)$ $\Kc_{n,-m}$ is the mirror image of $\Kc_{n,m}$; $iii)$ $\Kc_{-n,-m}$ is $\Kc_{n,m}$ with the opposite orientation. }
The unknot instead corresponds to the equivalence class of pairs $(n,m)$ with either $n$ or $m$ (or both) equal to $1$; with a slight abuse of notation we will chose the unknot representative as  $\Kc_{1,0}$. 
The first few non-trivial torus knots are listed in Figure \fig{torusknots}. 
\begin{figure}
    \centering
    \includegraphics[width=.8\textwidth]{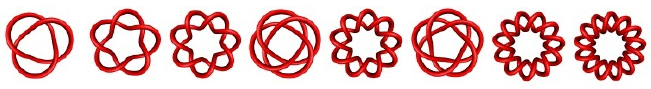}
    \caption{The first few non-trivial torus knots as told by their crossing number, namely $\Kc_{3,2}$ (the \emph{trefoil}), $\Kc_{5,2}$, $\Kc_{7,2}$, $\Kc_{4,3}$, $\Kc_{9,2}$, $\Kc_{5,3}$, $\Kc_{11,2}$, and $\Kc_{13,2}$ (figure taken from \cite{knotplotsite}). We recall that the crossing number of a $(n,m)$ torus knot is $\text{cr}(\Kc_{n,m})=\text{min}\big\{(n-1)m,(m-1)n\big\}=(m-1)n$.}
    \label{fig:torusknots}
\end{figure}

The pair of coprime numbers $(n,m)$ defines the family of diffeomorphisms
\be\label{Unm}
U_{n,m} =
\begin{pmatrix}
 n & \gamma \\
 m & \delta
\end{pmatrix}
\ee
that transform the unknot to $\Kc_{n,m}$, i.e., $\Kc_{n,m}=U_{n,m}\,\Kc_{1,0}$ since $U_{n,m}\left(\begin{smallmatrix} 1\\0 \end{smallmatrix}\right)=\left(\begin{smallmatrix} n\\m \end{smallmatrix}\right)$. Here, $\gamma,\delta$ are constrained by the unit determinant condition $|n\delta-m\gamma|=1$. 

In a generic situation the representation of a torus diffeomorphism of the form \eq{Unm} on the torus Hilbert space, characterized by the operator $\Uc_{n,m}$, defines a new quantum state
\begin{equation}\label{jnm}
    |j_{n,m}\rangle \equiv \Uc_{n,m}|j\rangle = \sum_i\big(\Uc_{n,m}\big)_{ji}|i\rangle\,.
\end{equation}
Non-trivial relations for the operator $\Uc_{n,m}$ that depend on the chosen representation will be discussed in section~\ref{sec:ambiguity}; these will reduce the size of the space of states (\ref{jnm}) compared to the corresponding knot space, but for the moment let us assume that such relations are not present. The interpretation of $|j_{n,m}\rangle$ is clear: it corresponds to a solid torus with a Wilson loop inserted along the torus knot $\Kc_{n,m}$, i.e., $W_{\Rc^*_j}(\Kc_{n,m})$, in contrast to $W_{\Rc^*_j}(\Kc_{1,0})$ that defines the original $|j\rangle$ (see Figure \fig{basis}). Notice that the unitary nature of $\Uc_{n,m}$ ensures that the states $\big\{|j_{n,m}\rangle\big\}$ are also orthonormal. In other words, a diffeomorphism on the torus amounts to a change of basis in $\Hc(T^2;G,k)$. We shall refer to $\big\{|j_{n,m}\rangle\big\}$ as the torus knot basis, as opposed to the unknot basis $\big\{|j\rangle\big\}$.

\begin{figure}[ht]
\centering
    \subfigure[]{\includegraphics[width=.3\textwidth]{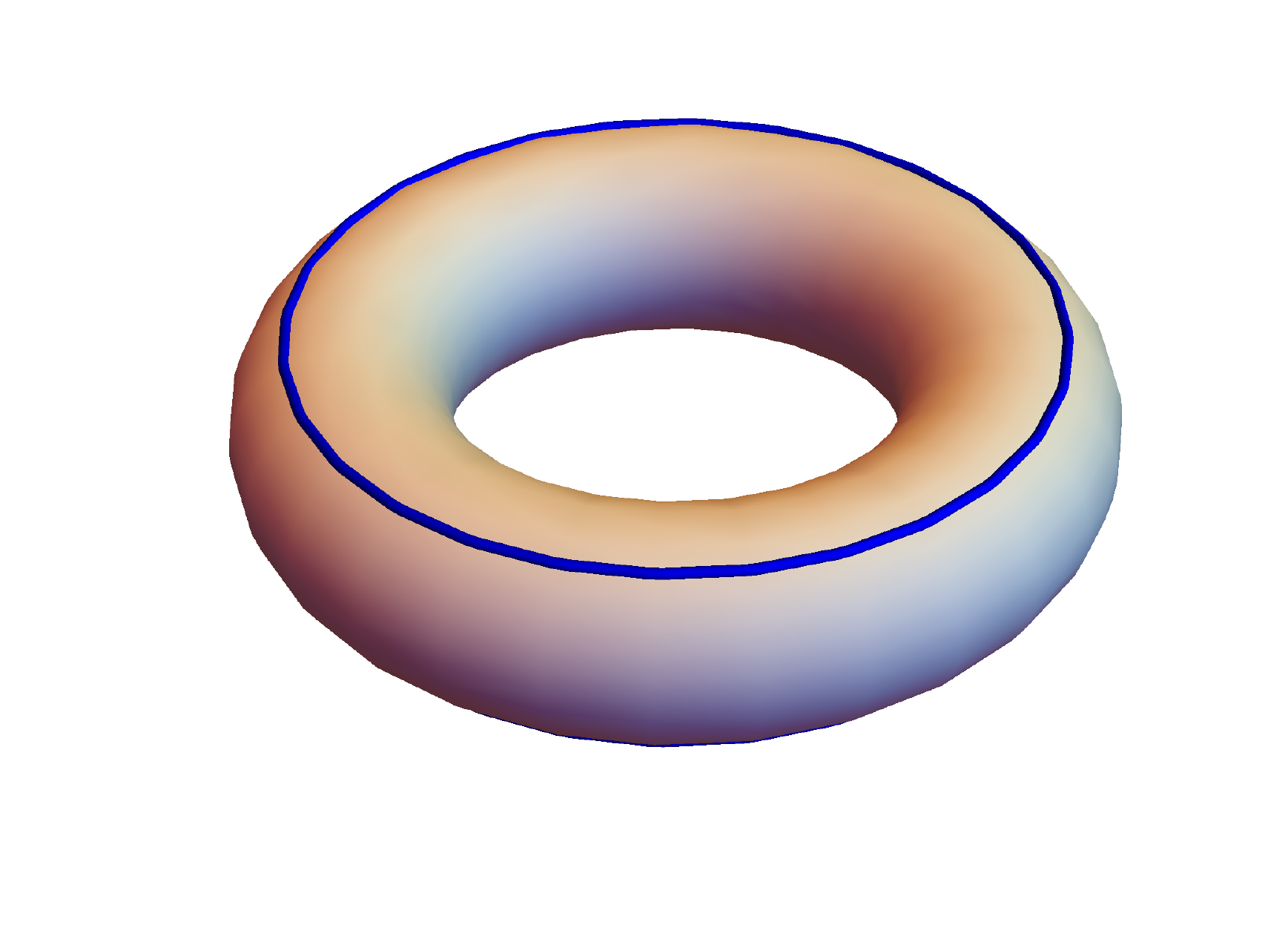}}\qquad\qquad
    \subfigure[]{\includegraphics[width=.3\textwidth]{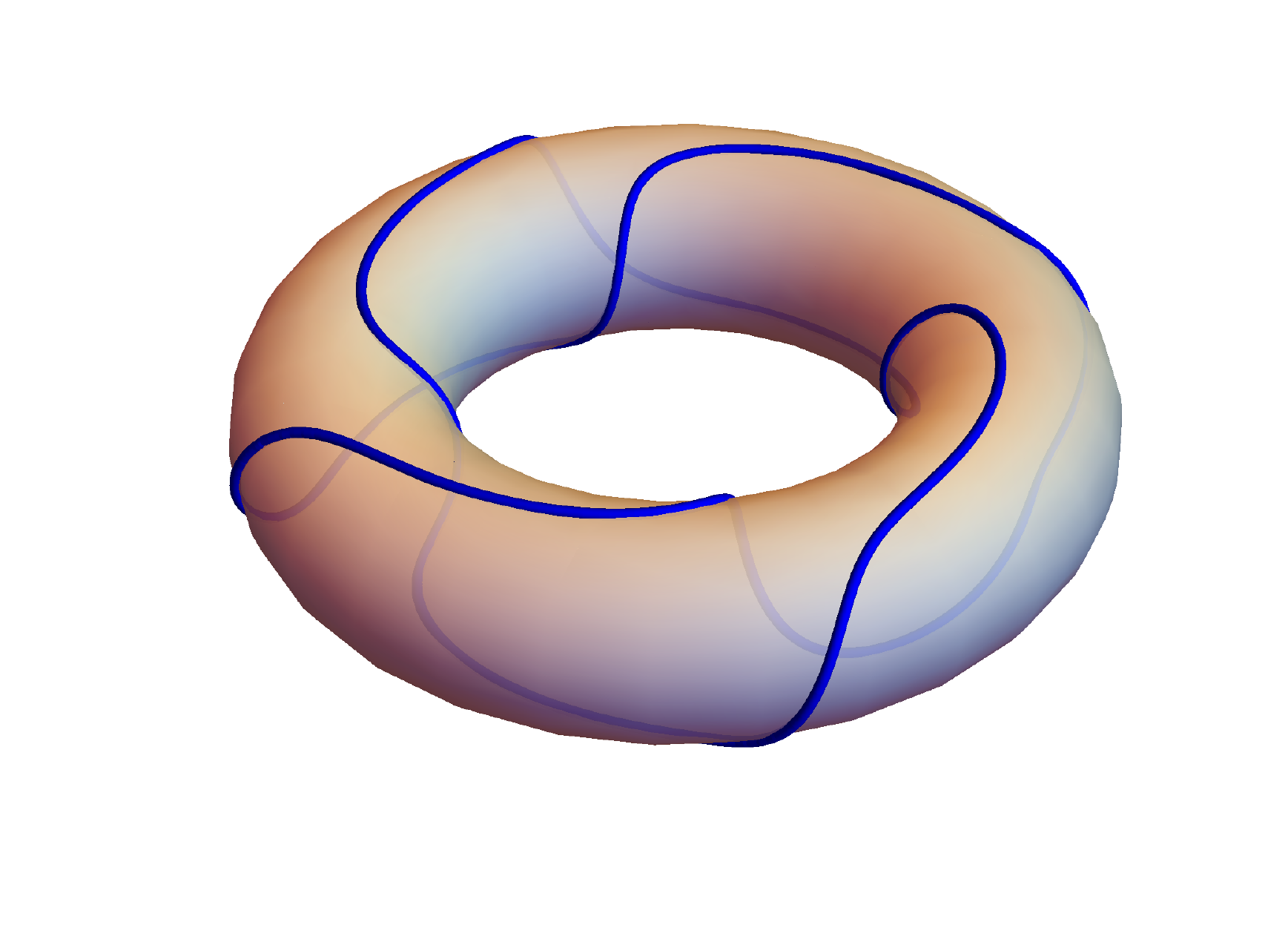}}
    \caption{The action \eq{jnm} of $\Uc_{n,m}$ on the basis vectors corresponding to a given representation $\overline{\Rc}_j$. It transforms the circular Wilson loop $|j\rangle$ depicted in (a) into the $(n,m)$ torus knot $|j_{n,m}\rangle$ illustrated in (b).}
    \label{fig:basis}
\end{figure}

The action of $\Uc_{n,m}$ on a generic torus knot complement state \eq{ketK} is less trivial,
\begin{align}\label{UnmonK}
    \Uc_{n,m}|\Kc_{p,q}\rangle &= \sum_{j}\psi_j(\Kc_{p,q})|j_{n,m}\rangle \notag\\
    &= \sum_{i}\bigg(\sum_j{\psi}_j(\Kc_{p,q})\big(\Uc_{n,m}\big)_{ji}\bigg)|i\rangle \,,
\end{align}
which is itself a different torus knot complement state, say $|\Kc_{r,s}\rangle$. This can be seen using the explicit operator description of torus knots in Chern-Simons theory in terms of which the natural action of $P\SL$ becomes manifest (see \cite{Labastida:1990bt} for details). Namely, an arbitrary torus knot Wilson loop $W_\Rc(\Kc_{p,q})$ can be obtained by acting with the \emph{torus knot operator} $\mathbf{W}_\Rc^{(p,q)}$ on the empty solid torus (the \lq\lq vacuum state\rq\rq), and modular transformations $U$ (with unitary representation $\Uc$) map these operators between themselves, that is, $\Uc^{-1}\mathbf{W}_\Rc^{(p,q)}\Uc=\mathbf{W}_\Rc^{(r,s)}$ with $\left(\begin{smallmatrix} r\\s \end{smallmatrix}\right)=U\left(\begin{smallmatrix} p\\q \end{smallmatrix}\right)$. It is clear from this construction that all torus knots can be obtained from the unknot (created by $\mathbf{W}_\Rc^{(1,0)}$) by an appropriate modular transformation, namely \eq{Unm}. In other words, $\Uc_{n,m}$ maps the unknot-complement state to the $\Kc_{n,m}$-complement, 
\begin{align}\label{Knm}
    \Uc_{n,m}|\Kc_{1,0}\rangle=|\Kc_{n,m}\rangle\,.
\end{align}
Some additional details about torus knot states are summarized in Appendix~\ref{sec:appendix}.

\subsection{Modular group versus unitary representation}
\label{sec:ambiguity}

An important subtlety to note is that going from a $P\SL$ group word $U_{n,m}$ to its unitary representation $\Uc_{n,m}$ may reduce the word length, since generically there are additional constraints satisfied by the unitary matrices $\Sc$ and $\Tc$. For instance, one of the fundamental properties of the quantum  modular
representation is that $\Tc$ is a diagonal matrix of finite order ($\Tc^p=1$ for some $p$) \cite{2005math......1443B}. This is evident in the $U(1)_k$  representation~\eq{STU(1)}, where
\begin{equation}
\label{TconstraintU(1)}
\Tc^{2k} = 1\,
\end{equation}
and in the $SU(2)_k$ representation~\eq{STSU(2)}, which has
\begin{equation}
\label{eq:1}
\Tc^{4(k+2)} = 1.
\end{equation}
Consequently, if two distinct torus knots are related by a modular transformation whose unitary representation is trivial (for instance, $U=ST^{2k}S$ in the $U(1)_k$ case yields $\Uc=\Sc\Tc^{2k}\Sc=1$) the corresponding knot states are actually equivalent. Moreover, as the matrix elements of $\mathcal{S}$ and $\mathcal{T}$ are roots of unity depending on the level $k$, extra modular constraints are to be expected in the unitary representation.

We can be more precise about the extra constraints in the modular
representation as follows. Let $\rho : \SL \rightarrow GL(d,\mathbb{C})$ be a
unitary modular representation acting on the Hilbert space
$\mathcal{H}$ of a Rational CFT (RCFT), where $d = \dim
\mathcal{H}$. As first discussed in \cite{Coste1999} and later proved
in broader generality in \cite{Ng:2008ze}, the kernel of $\rho$ is a
congruence subgroup of $\SL$, which then implies that the image of
$\rho$ is a representation of a finite group. Therefore, the unitary
modular representation can be understood as a representation of the
finite quotient group $SL_{2}(N) \equiv \SL / \Gamma(N)$, where
$\Gamma(N)$ is the principal congruence subgroup of level $N$, which in our case is a function of $G$ and $k$
\cite{Coste1999}. Consequently, all the torus knot states will reduce
to only a finite number of inequivalent classes, whose representatives
are related to each other by elements of the quotient. 

Understanding the modular coset representatives is crucial to
calculate the exact $\mathcal{S}$ and $\mathcal{T}$ word of a quantum
knot state. However, as the constraints actually depend on the group
$G$ and level $k$, we refrain to discuss these representations in any
more detail in this paper. In the rest of the paper, we will discuss
topological knots in terms of words in the the modular group $\SL$.
As will be argued in the next section, topological knots can be
understood as semiclassical knot states in the large level limit
$k \rightarrow \infty$ of quantum knot states. In this sense, the
semiclassical analysis will provide an upper bound on the complexity of
the quantum knot states.

\section{Circuit Complexity of Torus Knots and Knot States}
\label{sec:comp}

\subsection{Minimal words and topological complexity}

Let $U_{n,m}$ denote a torus diffeomorphism yielding $\Kc_{n,m}$ from the unknot as in \eq{Unm}. As any $P\SL$ matrix, it can be decomposed in terms of the $S$ and $T$ (and $S^{-1},T^{-1}$) generators as a word of the form 
\begin{align}\label{Unmword}
U_{n,m}=T^{a_1}ST^{a_2}S\ldots T^{a_r}S\,,
\end{align}
where $a_i$ are integer numbers and negative powers  are to be understood as positive powers of the inverse matrix. Here we have already used $S^2=1$ to exclude higher powers of $S$. Of course this decomposition is not unique thanks to the further group relation $(ST)^3=1$ and, more importantly, to the fact that \eq{Unm} actually defines a whole family of $U_{n,m}$'s (parametrized, e.g., by $\delta$). As a result, the word length 
\begin{align}\label{lengthU}
 ||U_{n,m}|| = \sum_{i=1}^r(|a_i|+1)
\end{align}
may grow indefinitely. Our goal is to find the shortest of these words.

We begin by clarifying the class of numbers $\{a_i\}$ that can appear in \eq{Unmword}. We first notice that each factor $T^{a_i}S\equiv M_{a_i} =\left(\begin{smallmatrix}a_i &-1\\ 1 &0\end{smallmatrix}\right)$ acts on the torus modular parameter $\tau$ as the M\"obius transformation $M_{a_i}(\tau)=a_i-\tfrac{1}{\tau}$. The full word $U_{n,m}$ then acts on $\tau$ as a composition of these maps, i.e., $U_{n,m}(\tau)=\tfrac{n\tau+\gamma}{m\tau+\delta}=M_{a_1}\circ M_{a_2}\circ \cdots \circ M_{a_r}(\tau)$. At $\tau=\infty$ this gives 
\begin{align}\label{aiCFD}
\frac{n}{m} &= M_{a_1}\circ M_{a_2}\circ \cdots \circ M_{a_r}(\infty)\notag\\
&= a_1-\frac{1}{a_2-\frac{1}{\ddots-\frac{1}{a_r}}}\,,
\end{align}
showing that the set of allowed $a_i$ forms a continued fraction (CF)\footnote{For further connections between knots and continued fraction see, e.g., \cite{CONWAY1970329}.} decomposition of $\tfrac{n}{m}$. Using the standard \q{all plus}~notation for continued fractions
\begin{align}\label{CFdef}
 [b_1;b_2,\ldots,b_r] \equiv b_1+\frac{1}{b_2+\frac{1}{\ddots+\frac{1}{b_r}}}\,,
\end{align}
one can write
\begin{align}\label{nmCF}
\frac{n}{m} &= [b_1;b_2,\ldots,b_r] \qquad \text{with}\qquad b_i=(-1)^{i+1}a_i
\end{align}
(in other words, $b_i=a_i$ for odd $i$ and $b_i=-a_i$ for even $i$). 

It is crucial to note that for a fixed $\tfrac{n}{m}$ there are infinitely many different CF decompositions \eq{nmCF}, since a priori no restriction is put on the integers $a_i$. For instance, if we declare the $a_i$ to be strictly positive (meaning that the use of the $T^{-1}$ generator is forbidden) we can write $\tfrac{3}{2}=[2;-2]$ so that $U_{3,2}=T^2ST^2S=\left(\begin{smallmatrix}3 &-2\\ 2 &-1\end{smallmatrix}\right)$ is a word of length $6$ producing the trefoil knot; on the other hand, if negative $a_i$'s are allowed we could write $\tfrac{3}{2}=[1;2]$ and get the shorter word $U_{3,2}=TST^{-2}S=\left(\begin{smallmatrix}3 &1\\ 2 &1\end{smallmatrix}\right)$ of length $5$ that also does the job. More generally, equivalent CF expansions exist that do not even have the same number of terms (e.g., $\tfrac{3}{2}=[-2;1,-1,-2,-2]$). The shortest word corresponds to the particular choice(s) of CF that minimizes \eq{lengthU}, that is,
\begin{align}\label{minlengthU}
\big|\big|U^{\text{min}}_{n,m}\big|\big| &\equiv \min_{r,\{a_i\}}\,\sum_{i=1}^{r} \big(|a_i|+1\big) = \min_{[b_1;b_2,\ldots,b_r]}\,\sum_{i=1}^{r} \big(|b_i|+1\big)\,.
\end{align}
Notice that there is an interplay between $r$ (the total number of terms in the CF) and the absolute values of the CF coefficients themselves, which makes the minimization procedure tricky: we want a CF with not too many terms, while at the same time keeping the coefficients sufficiently small. 

In principle there could be many CFs yielding this minimal length (there is no reason why it should be unique). For our purposes it is enough to find one of them, since we are only interested in the value of $||U^{\text{min}}_{n,m}||$ itself. 
A possible strategy would be to ignore the fact that $r$ and $\{b_i\}$ are related and minimize first over the number of terms $r$ to later worry about the $b_i$. Continued fractions with the least number of terms (so-called \emph{geodesic continued fractions}) have been discussed in \cite{beardon2012}, where a prescription is given to construct a particular geodesic CF using the so-called \emph{ancestral path} from $\tfrac{n}{m}$ to $\infty$ on the Farey graph. The problem here is that in general there are multiple geodesic CFs\footnote{Namely, there are at most $F_r$ (the $r$-th Fibonacci number) geodesic CFs with value $x$, where $r$ is the minimal number of terms needed to expand the rational $x$ \cite{beardon2012}.} and it is not clear how to carry out the minimization over coefficients $b_i$ within this set of geodesic CFs (the exception here is when $|b_i|\ge3$ for all $i\ge2$, in which case the ancestral path CF constructed in \cite{beardon2012} is the unique geodesic one and therefore minimizing over $b_i$ is trivial). Even though we have compelling numerical evidence that this ancestral path CF indeed minimizes the length \eq{minlengthU}, here we adopt a different (simpler) strategy and prove the following:

\begin{proposition}
\label{prop1}
For $n>m>0$, the minimal word length \eq{minlengthU} is achieved for the regular continued fraction representation of $\tfrac{n}{m}$ (i.e., the one for which $b_i>0$ for all $i=1,\ldots,r$) with $b_r>1$. This representation is unique.\footnote{For an arbitrary (positive or negative) rational $x$ this does not work. Even though we shall not need it here, we conjecture that in this case $C=\sum_{i=1}^r(|b_i|+1)$ is minimized by the continued fraction $x=[b_1,\ldots,b_r]$ obtained by modifying the Euclidean algorithm as follows:
\begin{enumerate}
 \item[$i)$]{If $k\le x\le k+\tfrac{1}{2}$ for some integer $k$, then it is the standard Euclidean continued fraction of $x$;}
 \item[$ii)$]{If $k-\tfrac{1}{2}<x<k$ for some integer $k$, then it is given by $[-b_1^*,\ldots,-b_r^*]$, where $[b_1^*,\ldots,b_r^*]$ is the Euclidean continued fraction of $-x$.}
\end{enumerate}
E.g., for $x=-\tfrac{1}{4}$ the Euclidean continued fraction $[-1,1,3]$ gives $C=8$ while the modified one yields $[0,-4]$ which has $C=6$. The proof should parallel the one given in the text after noting that this modified continued fraction has all coefficients of the same sign (except perhaps the first). We thank Ian Short for help on that.} 
\end{proposition}

\noindent First of all, let us recall that the regular (or Euclidean) continued fraction representation $[b_1,\ldots,b_r]$ of a rational number $x$ is defined by the recurrence relation
\begin{align}\label{Euclid}
b_i &= \lfloor f_i\rfloor\notag,\\
\quad f_i&=\frac{1}{f_{i-1}-b_{i-1}}\quad (i=2,\ldots,r)\,,\notag\\
\qquad f_1 &= x,
\end{align}
where $\lfloor \cdot\rfloor$ is the floor function. The construction implements the Euclidean algorithm for finding the greatest common divisor of two integers $n$ and $m$. 
It is clear that its coefficients are all positive (except perhaps the first one, which vanishes when $0\le x<1$ or is negative when $x<0$). 

\begin{proof}[Proof of Proposition~\ref{prop1}] 
Suppose that $[c_1,\ldots,c_r]$ is a continued fraction decomposition of $\tfrac{n}{m}$ that minimizes \eq{minlengthU}. What we need to show is that we can modify this expansion without changing the length $||U^{\text{min}}_{n,m}||=\sum_{i=1}^{r} \big(|c_i|+1\big)$ in such a way that all the coefficients of the resulting continued fraction become positive. This can be done with the help of the following two identities
\begin{subequations}\label{moves}
\begin{align}
 [\cdots,c_i,c_{i+1},\cdots] &= [\cdots,c_i-1,1,-c_{i+1}-1,-(\cdots)] \label{move1}\\
 [\cdots,c_i,c_{i+1},\cdots] &= [\cdots,c_i+1,-1,-c_{i+1}+1,-(\cdots)] \label{move2}\,,
\end{align}
\end{subequations} 
where $-(\cdots)$ means that all the coefficients that previously appeared in $\cdots$ are to appear now with opposite sign. The fact that expressions \eq{moves} still give an expansion of the same $\tfrac{n}{m}$ can be seen using $[a_1,\ldots,a_j]=a_1+\tfrac{1}{[a_{2},\ldots,a_j]}$ after checking the simpler identities\footnote{In the language of $S$ and $T$ transformations these identities are simply a manifestation of the $(ST)^3=1$ group relation. E.g., the first one is equivalent to $T^aST^{-b}S=T^{a-1}(TST)T^{-b-1}S=T^{a-1}(ST^{-1}S)T^{-b-1}S$.} $[a,b]=[a-1,1,-b-1]$ and $[a,b]=[a+1,-1,-b+1]$. It is also straightforward to check that the move \eq{move1} preserves the length $||U^{\text{min}}_{n,m}||$ as long as $c_i>0$ and $c_{i+1}<0$, while \eq{move2} preserves the length for $c_i<0$ and $c_{i+1}>0$. 

There is a small subtlety in using \eq{moves} when either $c_i$ or $c_{i+1}$ take the values $\pm1$, in which case the identities may generate a vanishing coefficient. Whenever this happens we have to convention that the rearrangements $[\cdots,a,0,1,\cdots]\equiv[\cdots,a+1,\cdots]$ and $[\cdots,1,0,a,\cdots]\equiv[\cdots,1+a,\cdots]$ have been done (both of which are trivially checked). Then, with that in mind, we can do the following procedure:
\begin{enumerate}
 \item[$(i)$]{If $c_1>0$: whenever a negative coefficient $(c_{i+1}<0)$ occurs in $[c_1,\ldots,c_r]$ we eliminate it by applying the move \eq{move1}. By repeated application from the left to the right, we 
     eliminate all the negative coefficients and obtain a new expression $[b_1,\ldots,b_s]$ having only positive $b_i$'s;}
 \item[$(ii)$]{If $c_1\le0$: first apply the move \eq{move2} sufficiently many times on the first pair of coefficients such that the resulting continued fraction gets a positive term in the first position. Then repeat the procedure in $(i)$ to make all the remaining ones positive as well.}
\end{enumerate}
This procedure converts any minimal $[c_1,\ldots,c_r]$ into a continued fraction with positive coefficients $b_i$ without changing $||U^{\text{min}}_{n,m}||$. It remains to show that such a decomposition can always be chosen with $b_r>1$, and in this case it is unique. 

We first notice that there is an ambiguity in the definition of the last term since $b_r$ can always be written as $(b_r-1)+\tfrac{1}{1}$, namely $[b_1,\ldots,b_r]=[b_1,\ldots,b_r-1,1]$. We would like to show that if we define all positive continued fractions (except the trivial case of $\tfrac{1}{1}$) in such a way that $b_r>1$, then such a presentation is unique. To see that, consider two decompositions of the same $\tfrac{n}{m}$,
\be
b_1+\frac{1}{[b_2,\ldots,b_r]} \ =\ \tilde{b}_1+\frac{1}{[\tilde{b}_2,\ldots,\tilde{b}_s]}\,,
\ee
where $b_i$ and $\tilde{b}_i$ are strictly positive integers with $b_r>1$ and $\tilde{b}_s>1$. First notice that the latter fact means that $[b_2,\ldots,b_r]>1$ and $[\tilde{b}_2,\ldots,\tilde{b}_s]>1$, which implies that $\tfrac{1}{[b_2,\ldots,b_r]}$ and $\tfrac{1}{[\tilde{b}_2,\ldots,\tilde{b}_s]}$ are not integers and therefore the equality is only possible if $b_1=\tilde{b}_1$ and $[b_2,\ldots,b_r]=[\tilde{b}_2,\ldots,\tilde{b}_s]$. By recursively applying this argument to the fractions $[b_2,\ldots,b_r]$ and $[\tilde{b}_2,\ldots,\tilde{b}_s]$, one concludes that $s=r$ and $b_i=\tilde{b}_i$ for all $i$. This proves Proposition 
\ref{prop1}.
\end{proof}

For future reference, let us also leave stated here the alternative proposition, for which we have no proof at the moment (only supporting numerical evidence):
\begin{proposition}
\label{prop2}
The minimal word length \eq{minlengthU} is also achieved for the geodesic continued fraction constructed in \cite{beardon2012} based on the ancestral path between $\tfrac{n}{m}$ and $\infty$ in the Farey graph. 
\end{proposition} 
\noindent For all practical matters, a proof is not needed since all we need here is the value of $||U^{\text{min}}_{n,m}||$, which can be computed using the simpler prescription of Proposition \ref{prop1}.

\vspace{.2cm}

We are now ready to discuss the topological complexity of torus knots. Recall from the previous section that every torus diffeomorphism $U_{n,m}$ producing the $\Kc_{n,m}$ torus knot from the unknot can be decomposed into a sequence of $S$ and $T$ transformations of the form $U_{n,m}=T^{a_1}ST^{a_2}S\ldots T^{a_r}S$. This $P\SL$ word is not unique, and the topological complexity of $\Kc_{n,m}$ corresponds to the length of the shortest possible one, $C_{n,m}=||U^{\text{min}}_{n,m}||$, which is the minimal number of transformations required to produce the desired knot. It is clear from Proposition \ref{prop1} above that $U^\text{min}_{n,m}=T^{a_1}ST^{a_2}S\ldots T^{a_r}S$ with $a_i=(-1)^{i+1}b_i$ and $b_i$ the Euclidean CF coefficients, so that
 \begin{align}\label{Cnm}
     C_{n,m}= \sum_{i=1}^{r} b_i+r+|f|\,
 \end{align}
where each factor $T^{a_i}S$ contributes with $|a_i|+1=b_i+1$ and the dependence on the particular knot $\Kc_{n,m}$ is encoded in both $r$ and $\{b_i\}$. Here we have jumped a bit ahead and added the $+|f|$ contribution due to framing of the initial unknot ($f$ is the self-linking number), which we now make a brief pause to clarify. 

When discussing the optimal circuit above, we have implicitly assumed that the knot $\Kc_{n,m}$ has trivial framing ($f=0$), which is not in general the case since no natural choice of framing exists. This ambiguity in the choice of framing affects many physical quantities, including the complexity. Fortunately, the way in which it affects $C_{n,m}$ is very simple: a knot $\Kc$ with $f$ units of framing and its version with trivial framing are mapped into each other by the \emph{$f$-fold Dehn twist} generated by $T^f$. This can be immediately seen in the $U(1)_k$ case from the transformation rule \eq{framing} of the wave function under a shift of framing, where the phase factor $\exp(2\pi\ii f\,j^2/k)$ picked up by the state is nothing but the matrix element of $\Tc^f$ that represents this Dehn twist in the Hilbert space, where $\Tc$ is shown in \eq{STU(1)}. Therefore, if $U_{n,m}^\text{min}$ is the minimal word building the trivially framed $\Kc_{n,m}$, it is clear that $U_{n,m}^\text{min}T^f$ is the corresponding one for the framed knot. Since in principle nothing prevents $f$ from being negative, this explains the $|f|$ extra units of complexity appearing in \eq{Cnm}. Any further change of framing by $t$ units as in \eq{framing} will add or subtract $t$ units of complexity to $C_{n,m}$.

In Table \tab{Cnmexplicit} we illustrate expression \eq{Cnm} explicitly for a number of particular torus knots $\Kc_{n,m}$, namely the ones for which the continued fraction decomposition of $\frac{n}{m}$ has up to 4 coefficients. This allows us to easily identify families of knots having the same topological complexity. For instance, it is clear from the first line that $C_{1+b_1b_2,b_2}=C_{1+b_1b_2,b_1}$ for all $b_1,b_2>1$, with analogous conclusions obtained from any permutation of indices in the second and third lines. Similarly, one can compare elements from different lines and identify further knots that are equally complex. 
A detailed exploration of these symmetries reveals that the vast majority of knots have a quite moderate topological complexity, as shown in Figure \fig{cpxtdistr}. It also becomes evident that the complexity increases very slowly in comparison with the crossing number of the knot, which provides a good notion of how knotted the given knot is.

\begin{table}[htb]
\centering
\begin{tabular}{|c|c|c|}
\hline
$n$                 & $m$         & $C_{n,m}-|f|$ \\ \hline
$1+b_1b_2$              & $b_2$         & $b_1+b_2+2$     \\ \hline
$b_1+b_3(1+b_1b_2)$         & $1+b_3b_2$      & $b_1+b_2+b_3+3$   \\ \hline
$1+b_1b_2+b_4(b_1+b_3(1+b_1b_2))$ & $b_2+b_4(1+b_3b_2)$ & $b_1+b_2+b_3+b_4+4$ \\ \hline
\end{tabular}
\caption{The topological complexity \eq{Cnm} for the particular torus knots $\Kc_{n,m}$ for which the  regular continued fraction decomposition of $\frac{n}{m}$ has up to $4$ coefficients. The construction proceeds similarly for $r>4$. Recall that, according to our convention, $b_i\ge1$ for $i=1,\ldots,r-1$ while $b_r$ is strictly $>1$.}\label{tab:Cnmexplicit}
\end{table}

\begin{figure}[thb]
    \centering
    \includegraphics[width=.65\textwidth]{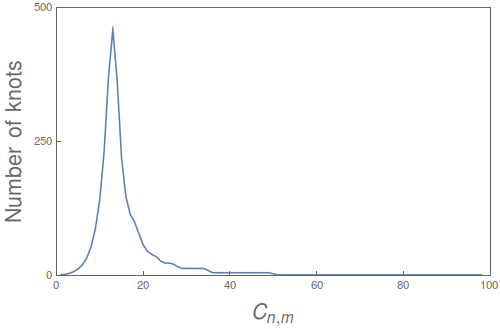}
    \caption{A statistical analysis of the topological complexity \eq{Cnm} for the first three thousand torus knots (as told by their crossing number), ignoring the framing contribution. While the crossing number of knots in the sample reaches up to $10^4$, $C_{n,m}$ increases much more slowly, reaching a maximum of $100$. In particular, the vast majority of states in the sample have moderate complexity (roughly between $10$ and $20$).}
    \label{fig:cpxtdistr}
\end{figure}

In Figure \fig{complexity-plot}, we analyze the behaviour of the topological complexity $C_{n,m}$  as a function of $n$ for different values of $m$.
\begin{figure}[ht]
  \centering
  \includegraphics[width=.65\textwidth]{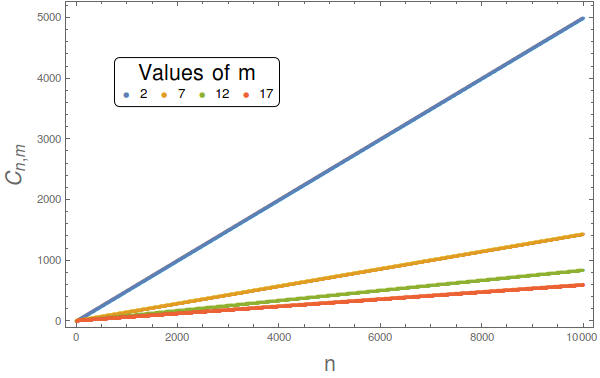}
  \caption{Topological complexity $C_{n,m}$ of torus knot states as a function of $n$ for different values of $m$. }
  \label{fig:complexity-plot}
\end{figure}
For large enough $n$, the topological complexity asymptotes to 
\begin{align}
\label{eq:growth}
    C_{n,m} \sim  \frac{n}{m}\,.
\end{align}
However, a more careful analysis shows that this an approximate behavior, since the apparent straight lines in the plot actually exhibit an internal structure with a periodic pattern of fluctuations that becomes more noticeable as $m$ increases. The situation is illustrated in Figure \fig{large-scale-m30} for $m=30$: the apparent single line is made of two discrete series of points that sit along two parallel lines. For other values of $m$ this pattern can be different. Even though we have no algebraic explanation for these patterns, we checked that they have the same asymptotic slope.
\begin{figure}[ht]
  \centering
  \includegraphics[width=.65\textwidth]{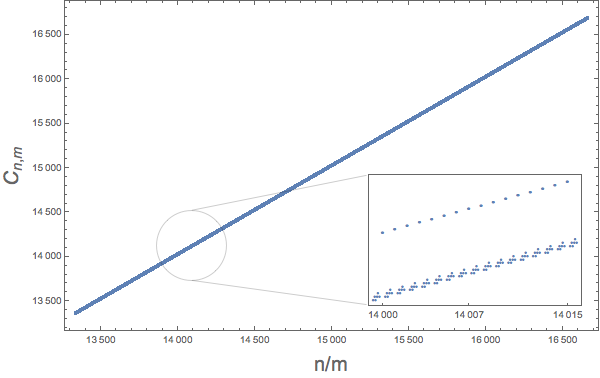}
  \caption{Large $n$ topological complexity for $m=30$. We zoom in a small interval to show its fine structure.}
  \label{fig:large-scale-m30}
\end{figure}

\subsection{Geometric interpretation}

Interestingly, the topological complexity \eq{Cnm} can be given a nice geometric interpretation in terms of the Farey tesselation of the upper half-plane \cite{Kulkarni}. This tesselation is a result of embedding the Farey graph $\mathcal{F}$ on the upper-half plane, whose vertices are the rational numbers (plus infinity) and whose edges join each pair of Farey neighbors (and only these). It is constructed as follows: one first draws vertical lines at each integer point $n=\frac{n}{1}\in\mathbb{Z}$, which reflects the fact that all integers are Farey neighbors of $\infty$; then joins pairs of adjacent points by a semicircle and generates a new rational point in between the two using the mediant formula $\tfrac{n_1}{m_1}\oplus\tfrac{n_2}{m_2}=\tfrac{n_1+n_2}{m_1+m_2}$. The procedure then follows recursively, and the result is shown by the black curves in Figure \fig{farey}. The resulting tiles are ideal triangles with vertices either at a rational number or at infinity, and edges along semicircles of different radii or along the vertical lines.\footnote{It can be proven that $i)$ two triangles are either equal or disjoint; $ii)$ every triangle is adjacent to exactly three other triangles; $iii)$ modular transformations map between different triangles, or, in other words, the Farey tesselation is invariant under $P\SL$.} One can also define the dual tree graph $\mathcal{F}^{*}$ of $\mathcal{F}$ by taking the baricenter of each ideal triangle as vertices and connecting them to the adjacent triangles, forming a trivalent tree.

The contribution from the sum of regular continued fraction coefficients $\sum_ib_i\equiv d_{n,m}$ in \eq{Cnm} is the \lq\lq geodesic distance\rq\rq~between the origin (representing the unknot) and the fraction $\frac{n}{m}$ that identifies $\Kc_{n,m}$. This geodesic distance $d_{n,m}$ is geometrically equivalent to the number of edges connecting $\frac{n}{m}$ to the origin along a semicircular path (a geodesic path in the upper-half plane) that goes through the ideal triangles, where the vertices correspond to intersections of this path with the curves in the Farey tesselation. See Figure \fig{farey}.

A related view of $\sum_ib_i$ can be given in terms of the Stern-Brocot tree of rational numbers. This is an infinite binary tree in which each vertex corresponds to a single positive rational number in its reduced form. It is constructed iteratively by starting at the zeroth level with the two extremal points $0=\frac{0}{1}$ and $\infty=\frac{1}{0}$ and, at each new level, introducing a new rational number in between every pair $\frac{n_1}{m_1}$ and $\frac{n_2}{m_2}$ of rationals in the previous level using the mediant formula mentioned above. The tree up to its fifth level is illustrated in Figure \fig{SBtree}. The sum $\sum_ib_i$ is the level (or \emph{depth}) in the Stern-Brocot tree of the node $\frac{n}{m}$, i.e., the number of edges connecting it to one of the root nodes on top. Let us emphasize here that two rationals having the same $\sum_ib_i$ does not imply that the corresponding knots have the same topological complexity, since \eq{Cnm} also contains a contribution from $r$. For instance, $\frac{4}{3}$ and $\frac{5}{3}$ both lie at the same depth in the Stern-Brocot tree ($\sum_ib_i=4$), but the corresponding torus knots $\Kc_{4,3}$ and $\Kc_{5,3}$ differ by one unit of complexity since the regular CF representation of $\frac{4}{3}=[1;3]$ has two coefficients while $\frac{5}{3}=[1;1,2]$ has three. Finally, we notice that the Stern-Brocot tree is isomorphic to the dual Farey graph $\mathcal{F}^{*}$. It has the mediants as vertices, but each mediant is simply connected to the baricenter of each its respective triangle. In this sense, the CF depth is equivalent to the number of triangles cut by the geodesic path of figure \ref{fig:farey}.

\begin{figure}[t]
    \centering    \includegraphics[width=.8\textwidth]{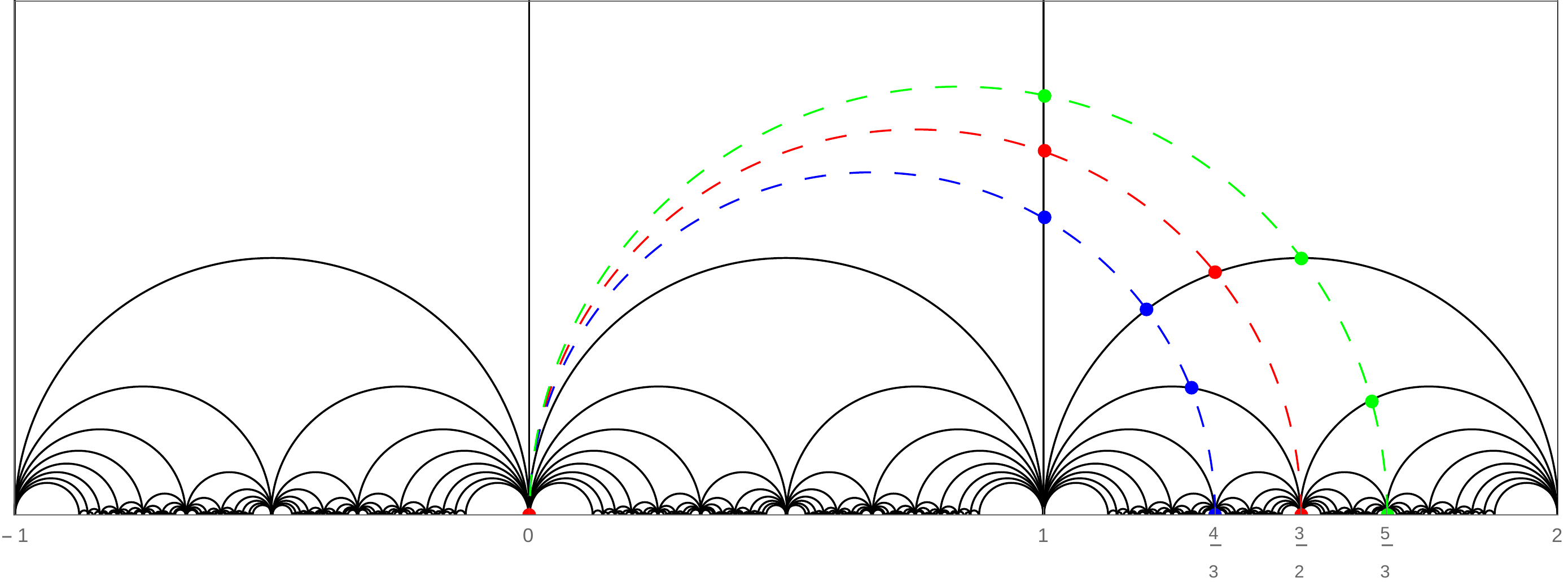}
    \caption{Geometric view of the sum of continued fraction coefficients contributing to \eq{Cnm} in the Farey tesselation of the hyperbolic plane (denoted by the black solid curves). 
    $\sum_ib_i\equiv d_{n,m}$ is the \lq\lq geodesic distance\rq\rq~between the origin (the unknot) and the fraction $\frac{n}{m}$ (the knot $\Kc_{n,m}$), i.e., the number of edges connecting them along a semicircular geodesic path that goes through the tiles. This is illustrated for $\Kc_{3,2},\Kc_{4,3}$, and $\Kc_{5,3}$ by the dashed red, blue, and green curves, respectively, which have $d_{3,2}=3$ and $d_{4,3}=d_{5,3}=4$.}
    \label{fig:farey}
\end{figure}

\begin{figure}[htb]
    \centering \includegraphics[width=.9\textwidth]{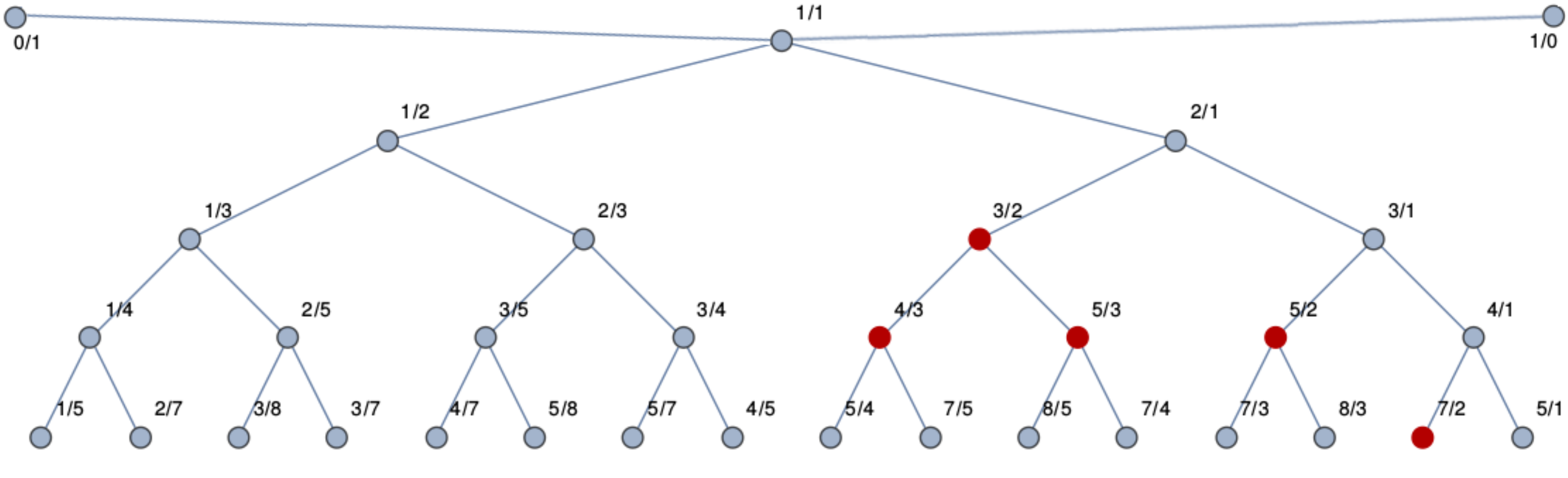}
    \caption{A view of the sum of continued fraction coefficients in the Stern-Brocot tree of positive rational numbers. Namely, $\sum_ib_i=d_{n,m}$ is the depth of the fraction $\frac{n}{m}$ in the tree. The red nodes highlight the locations of some of the torus knots illustrated in Figure \fig{torusknots}, namely $\Kc_{3,2},\Kc_{5,2},\Kc_{7,2},\Kc_{4,3}$, and $\Kc_{5,3}$ (the remaining ones appear in deeper levels not shown in the picture).}
    \label{fig:SBtree}
\end{figure}

The contribution from $r$ (the number of terms in the continued fraction) in \eq{Cnm} can also be seen geometrically in Figure \fig{farey}, but this time using \emph{paths along the Farey graph} itself (i.e., along the black solid curves as opposed to the dashed colored curves). The regular continued fraction $\tfrac{n}{m}=[b_1,\ldots,b_r]$ defines a connected path $\mathcal{P}_{n/m}=(\infty,C_1,\ldots,C_r=\tfrac{n}{m})$ in the Farey graph connecting $\infty$ to $\tfrac{n}{m}$ and whose vertices are the convergents $C_i\equiv[b_1,\ldots,b_i]$ ($i=1,\ldots,r$) of the continued fraction \cite{beardon2012}. Clearly, $r$ is the number of edges in this path.\footnote{Actually the same is true for any other integer continued fraction. The geodesic continued fractions (including the ancestral path in Proposition~\ref{prop2}) give the shortest of these paths.} For instance, the three examples of Figure \fig{farey} have $r_{4,3}=2,r_{3,2}=2$, and $r_{5,3}=3$ corresponding to the paths $\mathcal{P}_{4/3}=(\infty,1,\tfrac{4}{3}),\mathcal{P}_{3/2}=(\infty,1,\tfrac{3}{2})$, and $\mathcal{P}_{5/3}=(\infty,1,2,\tfrac{5}{3})$, respectively. It is clear that the paths above are not in general the ones that minimize $r$ (e.g., the path $(\infty,2,\tfrac{5}{3})$ is more efficient than $\mathcal{P}_{5/3}$), which reflects the fact that the regular continued fraction is not in general a geodesic continued fraction. 

Recently, the notion of circuit complexity, discrete groups and hyperbolic geometry has been discussed in \cite{Lin:2018cbk}. The main message there is that every finitely generated group $G$, with generating set $\mathcal{G}$, has a natural metric defined on its \emph{Cayley graph}. The vertices $g,h$ of this graph are elements of $G$ and we connect them by an edge if $gh^{-1} \in \mathcal{G}$.  The Cayley graph of the group of modular transformations on the torus, $P\SL$, is built from $S,T$ transformations (and its inverses) and its natural metric is exactly our topological complexity \eqref{Cnm}.

Summing up, the topological complexity of torus knots has two main geometric components (apart from the framing contribution $|f|$), one associated to the number of triangles, $\sum_ib_i$, and another to the number of edges of the Farey path from $ \infty$ to $\frac{n}{m}$, which is the CF size $r$. If we define $d_{\mathcal{F}}(a,b)$ and $d_{\mathcal{F}^{*}}(v_{a},v_{b})$ as the geodesic  metrics in the Farey graph and its dual tree respectively, where $a,b$ are rationals and $v_{a},v_{b}$ are ideal triangles with $a,b$ as mediants, it is clear that
\begin{align}
  \label{CayleyFarey}
    C_{n,m} \equiv d_{\text{Cayley}} =  d_{\text{Farey}^*}+d_{\text{Farey}}+ |f|,
\end{align}
where $d_{\text{Cayley}}$ is the word metric in the Cayley graph of
$S$ and $T$,
$d_{\text{Farey}^*}\equiv
d_{\mathcal{F}^*}\left(v_{1},v_{\frac{n}{m}}\right)+1 =\sum_ib_i $,
and
$d_{\text{Farey}} \equiv
d_{\mathcal{F}}\left(\infty,\frac{n}{m}\right)=r$. Notice that we have
to add one to the dual Farey metric as we need to
count the number of vertices, not just the number of edges.\footnote{For a tree graph, $V=E+1$, where $V$ is the number of vertices and $E$ is the number of edges.}

Now, let us set the framing to be zero. As we stated above, one
contribution to the topological complexity comes from the number of arcs in the Farey graph and another from the number of triangles along the
path. In our numerical tests, we observed an upper bound on
$d_{\mathcal{F}} \lesssim \log_{2}m+2$ for fixed $m$ up to $2000$,
which is consistent with bounds presented in
\cite{Schleimer2006} for paths in the Farey graph. This means that the
linear growth \eqref{eq:growth} for large $n$ must be only due to
$d_{\mathcal{F}^{*}}$, i.e., the number of triangles. This result
suggests a relation to the proposal of holographic subregion
complexity \cite{Alishahiha2015,Abt2018}, in which the complexity of a
subregion (in our case, the interval $[0,\frac{n}{m}]$) is
proportional to the volume bounded by a bulk geodesic. The normalized area below the geodesic in Figure~\fig{farey} grows as $n/2$ in the case $m=2$. However, despite the geometric similarity, our
holographic setup is more in the spirit of \cite{Manin2002},
as our geometric representation is for the moduli space of knot
states in Chern-Simons rather than the bulk spacetime itself.

\subsection{Complexity of Torus Knot States}

So far our analysis has been limited to the classical realm, namely to the topological complexity of the torus knots themselves as told by their topological properties encoded in the $S$ and $T$ torus diffeomorphisms needed to produce the knot. Let us now move on to discuss the circuit complexity of their corresponding quantum states in the torus Hilbert space of Chern-Simons theory with gauge group $G$ and level $k$. 

In Section 2 we mentioned that every torus diffeomorphism $U_{n,m}$
naturally defines a quantum circuit
$\Uc_{n,m}=\Tc^{a_1}\Sc\Tc^{a_2}\Sc\ldots\Tc^{a_r}\Sc$ based on $\Sc$
and $\Tc$ gates (the unitary representations of the $S$ and $T$
diffeomorphisms) acting on $\Hc(T^2;G,k)$. This quantum circuit can be
used to connect two different pairs of states on this Hilbert space,
namely the different basis vectors $|j_{n,m}\rangle$ and $|j\rangle$
(see \eq{jnm}) or, equivalently, the different torus knot complement
states $|\Kc_{n,m}\rangle$ and $|\Kc_{1,0}\rangle$ (see \eq{Knm}). Our
goal is to calculate the complexity of this circuit,
\begin{align}
  \Cc\big(|j_{n,m}\rangle,|j\rangle\big) = \Cc\big(|\Kc_{n,m}\rangle,|\Kc_{1,0}\rangle\big) \equiv \Cc_{n,m}\,,
\end{align}
which is the minimal number of gates required to generate the desired
unitary. We use the calligraphic $\Cc_{n,m}$ to distinguish this from the topological
complexity $C_{n,m}$ studied above. In other words, $\Cc_{n,m}$ is the length of the shortest
$\Sc\Tc$-word representation of $\Uc_{n,m}$ (the optimal circuit
$\Uc_{n,m}^\text{opt}$) in the unitary modular representation.

The naive expectation is that the optimal circuit is just the Hilbert
space representation of the minimal word $U^{\text{min}}_{n,m}$ in the
group manifold obtained in the previous section. However, this fails
in general for the reasons stressed in Section \ref{sec:ambiguity},
which boil down to the fact that the unitaries $\Sc,\Tc$ satisfy
further matrix relations beyond the $P\SL$ group relations
$\Sc^2=(\Sc\Tc)^3=1$. As a result, it follows that, in general,
$\Uc_{n,m}^\text{opt}\neq\Tc^{a_1}\Sc\ldots\Tc^{a_r}\Sc$ where
$a_i=(-1)^{i+1}b_i$ and $b_i$ are the regular continued fraction
coefficients of $\tfrac{n}{m}$ as in Proposition \ref{prop1}. For
instance, consider the cases of $U(1)_k$ and $SU(2)_k$ Chern-Simons,
where an immediate constraint is $\Tc^{p}=1$ with $p=2k$ and
$p=4(k+2)$, respectively; it is clear that, whenever any $a_i$
appearing in a group word
$U_{n,m}=T^{a_1}S\ldots T^{a_i}S\ldots T^{a_r}S$ is larger in modulus
than some multiple of $p$, that is $|a_i|\geq \ell p$ with
$\ell\in \mathbb{N}$, the Hilbert space representation of this
$U_{n,m}$ will be equivalent to
$\Uc_{n,m}=\Tc^{a_1}\Sc\Tc^{a_2}\Sc\ldots\Tc^{(-1)^{i+1}(|a_i|-\ell
  p)}\Sc\ldots\Tc^{a_r}\Sc$, which contains $\ell p$ less
generators.\footnote{To avoid any confusion, let us stress here that the resulting list of coefficients $[\tilde{a}_1;\tilde{a}_2,\ldots,\tilde{a}_r]$ no longer corresponds to a CF decomposition of the original fraction $\tfrac{n}{m}$.} Thus if $U_{n,m}$ is the minimal word
$U^{\text{min}}_{n,m}$, its Hilbert space representation
$\Uc^{\text{min}}_{n,m}$ (which is not necessarily the same as the
optimal circuit $\Uc^{\text{opt}}_{n,m}$) in general will be
reducible, and the circuit complexity $\Cc_{n,m}$ will end up being lower than the topological
complexity $C_{n,m}=||U^{\text{min}}_{n,m}||$. Therefore, all one can
say for generic gauge group $G$ and level $k$, without considering the extra constraints of the quantum representation, is that the
circuit complexity $\Cc_{n,m}=||\Uc^\text{opt}_{n,m}||$ has an upper
bound given by the topological complexity \eq{Cnm}, i.e.,
\begin{align}
\label{Cbound}
    \Cc_{n,m} \le C_{n,m}\,.
\end{align}

Fortunately, in the semiclassical limit $k\to\infty$ of Chern-Simons
theory, the above-mentioned subtleties disappear and the torus knot
states actually saturate this upper bound. Namely, in this case the
Hilbert space analysis parallels the $P\SL$ minimal word problem
solved by Proposition \ref{prop1}, the optimal circuit is simply
$\Uc_{n,m}^\text{opt}=\Tc^{a_1}\Sc\Tc^{a_2}\Sc\ldots\Tc^{a_r}\Sc$ and
the complexity of torus knot states $\Cc_{n,m}$ coincides with the
topological complexity of $\Kc_{n,m}$. In particular, in semiclassical
Chern-Simons theory all the plots, tables, and the geometric
interpretation of the topological complexity shown in previous section
extend to $\Cc_{n,m}$ as well. The equivalence between the quantum and
classical modular group representations in the large $k$ limit has been proved for $SU(2)$
\cite{Freedman2006,Marche2008} and for $SU(N)$
\cite{Andersen2006}. Moreover, the asymptotic limit of torus knot
states has been studied in \cite{Charles2015a,Charles2015}. The
conclusion is that torus knot states in this case are classified in the same way as the topological knots.

\section{Generalizations}
\label{sec:general}

\subsection{Connected sums of torus knots}
\label{sec:links}

The results are easily generalized for multi-component \emph{links} of the type
\begin{align}\label{Lnm}
 \Lc_{(n_1,m_1),\ldots,(n_N,m_N)}\equiv\Kc_{n_1,m_1} + 2^2_1 + \Kc_{n_2,m_2} + 2^2_1 + \cdots + \Kc_{n_L,m_L}\,,
\end{align}
where the plus sign here indicates the operation of connected sum of knots and $2^2_1$ denotes the Hopf link (in Rolfsen's notation), namely the simplest possible two-component link made of two unknots linked exactly once. In words, what this means is that $\Lc_{(n_1,m_1),\ldots,(n_N,m_N)}$ is the $N$-component link obtained by sequentially \lq\lq Hopf-linking\rq\rq~the torus knots $\Kc_{n_i,m_i}$, as illustrated in Figure \fig{linkedknots}. The simplest representative is $\Lc_{(1,0),\ldots,(1,0)}$, which is a $N$-component generalization of the Hopf link (it reduces to the standard Hopf link when $N=2$) obtained by a chain of trivial knots linked with unit linking numbers. 
It is important to emphasize that the class of links above is very special due to its simple linking pattern. In particular, it does \emph{not} include the torus links $(n,m)$, which are also composed by multiple torus knots (namely, $N=\text{gcd}(n,m)$ copies of $\Kc_{\frac{n}{N},\frac{m}{N}}$) but whose linking structure is more intricate.

\begin{figure}[ht]
    \centering
    \includegraphics[width=.4\textwidth]{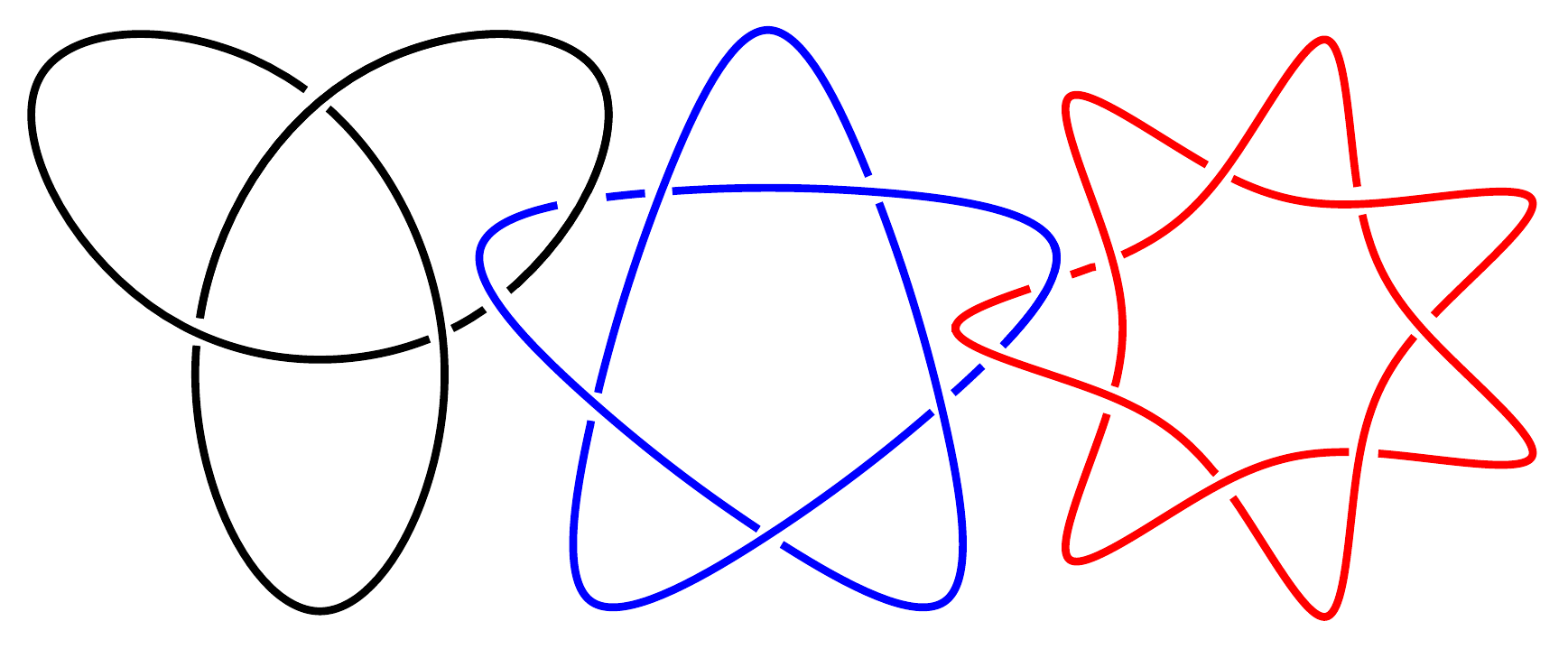}
    \caption{The link $\Lc_{(3,2),(5,2),(7,2)}$ obtained by \lq\lq Hopf-linking\rq\rq~$\Kc_{3,2}$ (black), $\Kc_{5,2}$ (blue), and $\Kc_{7,2}$ (red).}
    \label{fig:linkedknots}
\end{figure}

The story is then analogous to the one of knots in the previous section (we refer the reader to \cite{Balasubramanian:2016sro} for details). One can construct a $3$-manifold $\Mc$ that is the link-complement of \eq{Lnm} in $S^3$ and do the Chern-Simons path integral on it. This defines a state $\big|\Lc_{(n_1,m_1),\ldots,(n_N,m_N)}\big\rangle$ on the boundary of $\Mc$, which now consists in the disjoint union of $N$ tori, i.e., $\partial\Mc=\coprod_{i=1}^N T^2$. The corresponding Hilbert space is $\Hc(T^2;G,k)^{\otimes N}$, for which a natural basis is given by $|i_1,\ldots i_N\rangle$ and $|i_\alpha\rangle$ corresponds to a circular Wilson loop in the (integrable) representation $\Rc_i$ inside the $\alpha$-th torus. As in \eq{psi(K)}, the wave function of a link complement state in this basis is nothing but the Chern-Simons invariant \eq{W(L)} of the corresponding link. 

The $S$ and $T$ diffeomorphisms of the $\alpha$-th torus transform the unknot inside this particular torus into an arbitrary torus knot $\Kc_{n_\alpha,m_\alpha}$, and the link \eq{Lnm} is constructed from the generalized Hopf link by the word $U_{(n_1,m_1),\ldots,(n_N,m_N)}=\prod_\alpha U_{n_\alpha,m_\alpha}$ with $U_{n_\alpha,m_\alpha}$ an $ST$-word as in the previous section. As before, these $S$ and $T$ transformations naturally define unitary operators acting on the $\alpha$-th single-torus Hilbert space. These operators are again subject to representation constraints, contained in the kernel of the map from $P\SL$ to the space of operators on $\Hc(T^2;G,k)$. Modulo the identification, the operators construct the state corresponding to the link \eq{Lnm} from the one corresponding to the generalized Hopf-link using a unitary transformation $\Uc_{(n_1,m_1),\ldots,(n_N,m_N)}$, i.e.,
\begin{align}
\big|\Lc_{(n_1,m_1),\ldots,(n_N,m_N)}\big\rangle=\Uc_{(n_1,m_1),\ldots,(n_N,m_N)}\big|\Lc_{(1,0),\ldots,(1,0)}\big\rangle\,.
\end{align}
This circuit clearly factorizes into a product of circuits of the type \eq{Knm}, one for each component knot $\Kc_{n_\alpha,m_\alpha}$, 
\begin{align}
    \Uc_{(n_1,m_1),\ldots,(n_N,m_N)} = \Uc_{n_1,m_1}\ldots\Uc_{n_N,m_N}\,,
\end{align}
and therefore has a complexity given by the sum of the complexities \eq{Cnm} for each of these components, namely 
\begin{align}
    \Cc_{(n_1,m_1),\ldots,(n_N,m_N)} = \sum_{\alpha=1}^N\Cc_{n_\alpha,m_\alpha}\,.
\end{align}
Each component $\Cc_{n_\alpha,m_\alpha}$ is again bounded from above by the topological complexity $C_{n_\alpha,m_\alpha}$ of the corresponding knot $\Kc_{n_\alpha,m_\alpha}$, that is
\begin{align}
    \Cc_{n_\alpha,m_\alpha}\le C_{n_\alpha,m_\alpha}
\end{align}
with
\begin{align}
    C_{n_\alpha,m_\alpha} = \sum_{i_{\alpha}=1}^{r_{\alpha}} b_{i_{\alpha}}+r_\alpha+|f_\alpha|\,,
\end{align}
where $b_{i_\alpha}$ and $r_{\alpha}$ are the coefficients and size of the regular continued fraction decomposition of $\tfrac{n_\alpha}{m_\alpha}$ and $f_{\alpha}$ is the framing contribution.

\subsection{Rational knots and links}
\label{sec:2bridge}

The notions of complexity for torus knots and torus knot states studied above can be extended to other interesting examples of TQFT Hilbert spaces if one uses the presentation of knots and links as closures of braids. In this case, we can make sense of circuit complexity as the length of the minimal word of the representations of the braid group generators, while the topological complexity corresponds to the minimal word in the braid group itself. For an alternative definition of the complexity of braids, see \cite{Dynnikov2007}. 

This requires going beyond the torus Hilbert spaces studied in Section \Sec{CS}, so let us start by reminding what are the Hilbert spaces associated with 2-spheres in Chern-Simons~\cite{Witten:1988hf}. First, the Hilbert space $\Hc(S^2;G,k)$ of  Chern-Simons theory in a 3-manifold with a $S^2$ boundary is one-dimensional. This is not a very interesting example, since all the vectors in such a space differ only by a phase factor. In order to have a non-trivial Hilbert space one has to consider spheres with removed points (\emph{punctures}) which correspond to the endpoints of Wilson lines. Hence punctures are alike non-dynamical heavy charged particles and they carry representations of the group $G$ admissible by the value of $k$.

One cannot consider a sphere with a single puncture, since there would be no place for the Wilson line emanating from this point to end. For two punctures one can have a Wilson line connecting them, so the two endpoints should carry conjugated representations of $G$. Consequently, even in this case the Hilbert space remains one-dimensional. 

In order to have a non-trivial Hilbert space, one has to consider a sphere with at least three punctures. The representations carried by these punctures should be compatible, in the sense that their tensor product should contain a trivial representation, $R_\emptyset\in R_1\otimes R_2\otimes R_3$. One can think of this as the interaction of particles, in which two particles can fuse to produce a third one. The multiplicity of the trivial representation in the tensor product is then known as the fusion number $N_{R_1R_2}^{\overline{R}_3}$. These numbers also appear in the chiral algebra of the $SU(2)_k$ WZW theory \cite{VERLINDE1988360}. They define the dimension of the resulting Hilbert space,
\be
\dim\Hc\big(S^2\backslash\{R_1,R_2,R_3\};G,k\big)=N_{R_1R_2}^{\overline{R_3}}\,.
\ee

Below we will consider states in the Hilbert spaces of spheres with four punctures. The representations of the four points are subject to the condition $R_\emptyset\in R_1\otimes R_2\otimes R_3\otimes R_4$. The dimension of the Hilbert space counts all the compatible ways of pairwise fusing the representations. We can represent the basis vectors by the diagrams
\be
\label{cblocks}
|R_i\rangle \ = \ \begin{array}{c}
\includegraphics[scale=0.4]{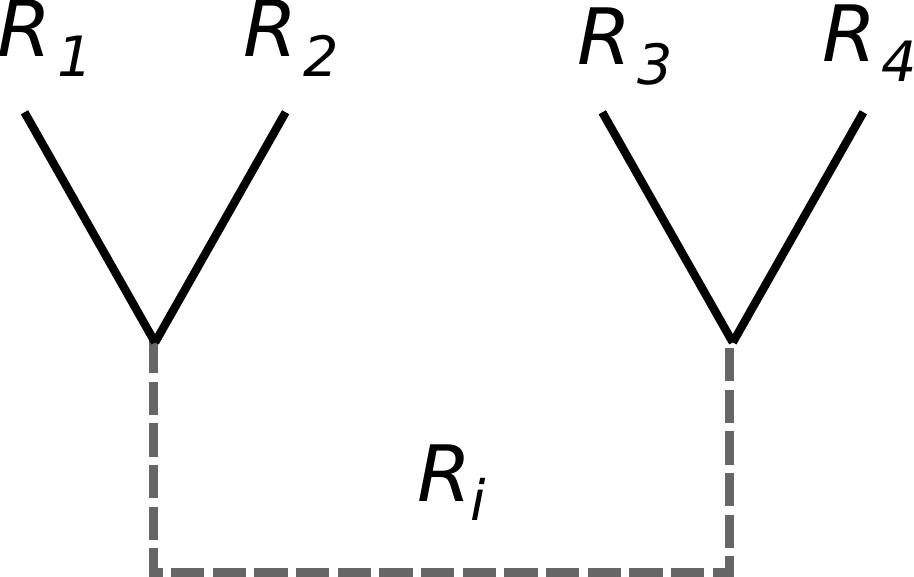}     
\end{array}
\qquad \text{or}
\qquad 
|\widetilde{R_j}\rangle \ = \ \begin{array}{c}
\includegraphics[scale=0.4]{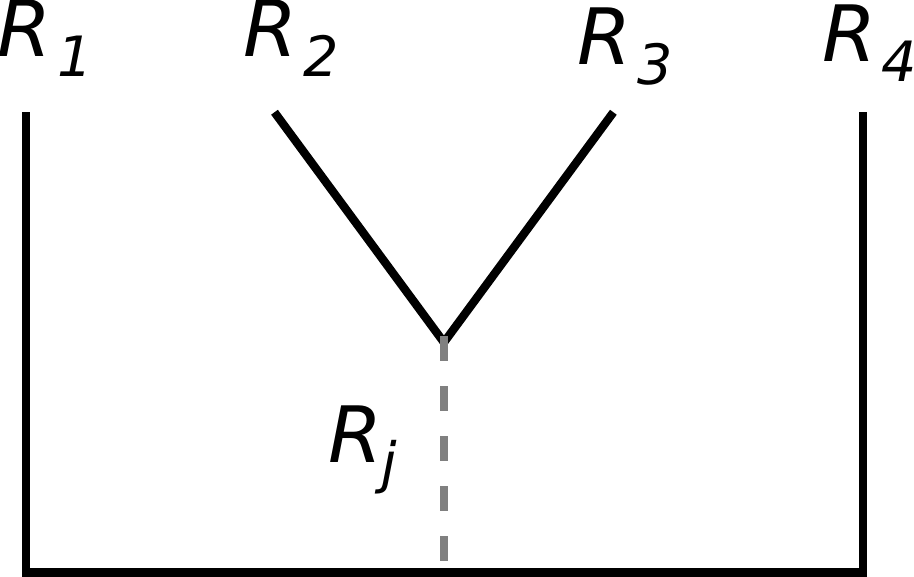}     
\end{array}\,.
\ee
Here $|R_i\rangle$ and $|\widetilde{R_j}\rangle$ correspond to two possible choices of the basis labeled by representations $R_i$, or $R_j$ (denoted by dashed lines) that can appear in the $s$ or $t$ fusion channels. To fix a reference state, let us choose $R_1=R_3=R$ and $R_2=R_4=\bar{R}$ in the left diagram.  We choose the reference state $|\Psi_\text{R}\rangle$ to correspond to trivial $R_i=\emptyset$, namely
\be
\label{Rstate}
|\Psi_\text{R}\rangle \ = \ \begin{array}{c}
\includegraphics[scale=0.4]{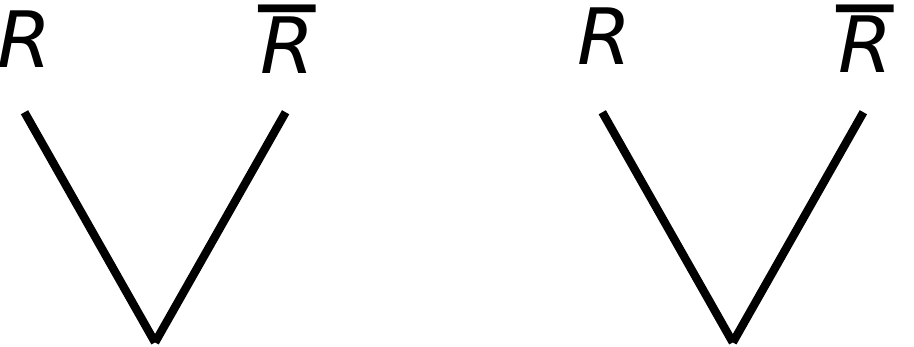}     
\end{array}\,,
\ee
where we understand that the line in a trivial representation is the same as no line at all.

The Hilbert space has a well-defined action of the braid group on it, which permutes the punctures on the sphere. The action of the braid group can be illustrated by a concatenation of a braid and the reference state:
\be\label{braidcircuit}
|\Psi_\text{T}\rangle \ = U\,|\Psi_\text{R}\rangle \ = \ \begin{array}{c}
\includegraphics[scale=0.4,angle=90]{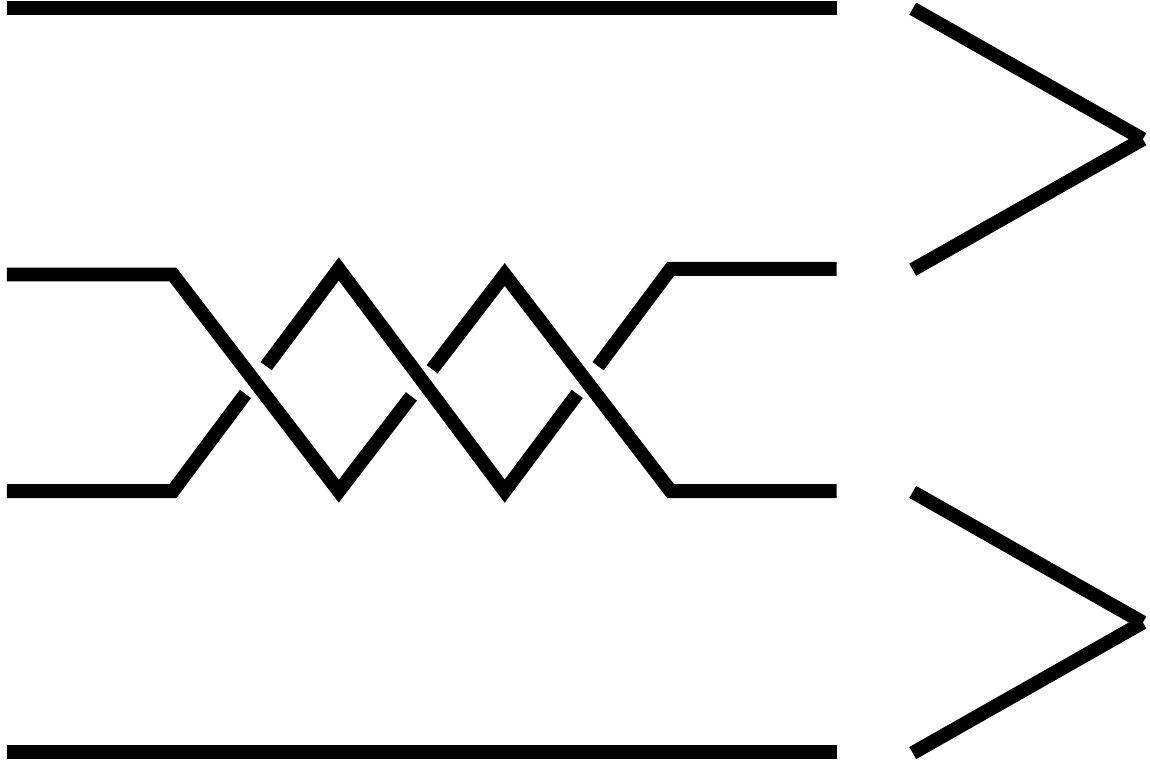}     
\end{array}
\ee
This illustrates the kind of target states one might be interested in. Note that, in this particular example, the state can be directly associated to a knot by closing a braid on the left by $\langle \Psi_\text{R}|$. The target state in the example above is represented by the trefoil (torus) knot.

The Hilbert space associated with four-punctured spheres is particularly suitable to the discussion of \emph{rational} (or \emph{2-bridge}) knots and links \cite{Murasugi,Duzhin:2010},\footnote{The case of generic knots and their minimal braid words will be addressed in a forthcoming work \cite{ongoing}.} for which many problems in knot theory can be completely solved. They are encoded by Artin's braid group of three elements, $B_3$. The crucial property for our purposes is the fact that $P\SL$ furnishes unitary representations of $B_3$, which allows us to take advantage of the results of the previous section. Indeed, one can easily check that
\begin{equation}
    \sigma_1= T\,,\qquad
    \sigma_2=STS\,
\end{equation}
satisfy the braid group relation 
\begin{equation}
    \sigma_1\sigma_2\sigma_1=\sigma_2\sigma_1\sigma_2\,.
\end{equation}
Together with the representation structure, this construction inherits the constraints, which potentially identify sets of states created by the braid group operators for the same reasons as discussed in Section \ref{sec:ambiguity}. Hence, in the quantum case, the complexity provides an upper bound, as in equation~(\ref{Cbound}), which is saturated by the topological complexity in the semiclassical limit $k\to\infty$.

Rational knots and links can be defined as closures of $B_3$ braids that are trivially embedded in $B_4$ (that is, the $\sigma_3$ generator of $B_4$ is never used in the braid word). There are two possible ways to close either end of a $B_4$ braid by connecting the strands pairwise: by the state $|\Psi_\text{R}\rangle$ in equation~\eq{Rstate}, or by an analogous state corresponding to trivial $R_j$ in the right diagram of equation~\eq{cblocks}. Consequently, there are two \q{bridges} connecting the strands and the obtained knots and links are also called 2-bridge.

This class of knots $L_{n,m}$ is again labelled by two coprime integers $(n,m)$ satisfying $n>0$ and $\big|\tfrac{n}{m}\big|\le1$ and can also be associated with continued fractions~\cite{Murasugi,Duzhin:2010}. For a fraction $\tfrac{n}{m}=[a_1,\ldots,a_r]$ with arbitrary integer coefficients as in \eq{CFdef}, one constructs the braid word $U_{n,m}=\sigma_2^{a_1}\sigma_1^{a_2}\sigma_2^{a_3}\cdots$. On the right end of the braid one always closes the strands by reference state~\eq{Rstate}, see Figure~\fig{close}. Notice that in the case of even $r$ the braiding starts with $\sigma_1^{a_r}$ which (within the convention above) is just a framing contribution. The choice of closure on the left end is uniquely fixed depending on the value of $a_1$: if $a_1=0$, meaning that the leftmost braiding operation is given by $\sigma_1$, then one must close the braid as in \fig{close}(b); otherwise, the leftmost operation is given by $\sigma_2$ and one must close the braid as in \fig{close}(a). This associates a unique knot diagram to every braid word. To be precise, a nice additional feature of the rational family is that apart from knots it also contains links: given $\frac{n}{m}$, the corresponding diagram is a knot if $m$ is odd, or a link when $m$ is even.
\begin{figure}[ht]
\centering
    \subfigure[]{\includegraphics[angle=90,width=.3\textwidth]{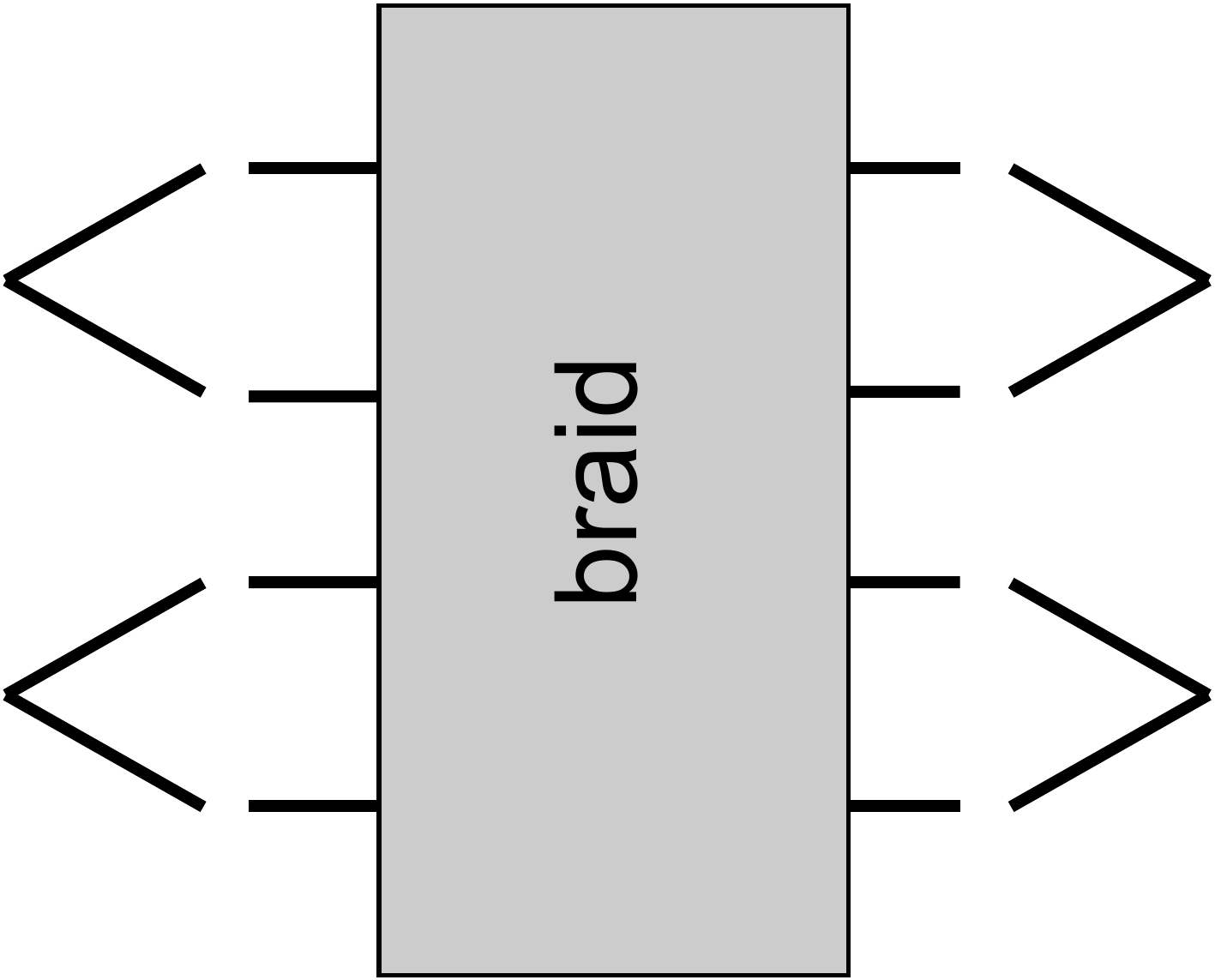}}\qquad\qquad
    \subfigure[]{\includegraphics[angle=90,width=.3\textwidth]{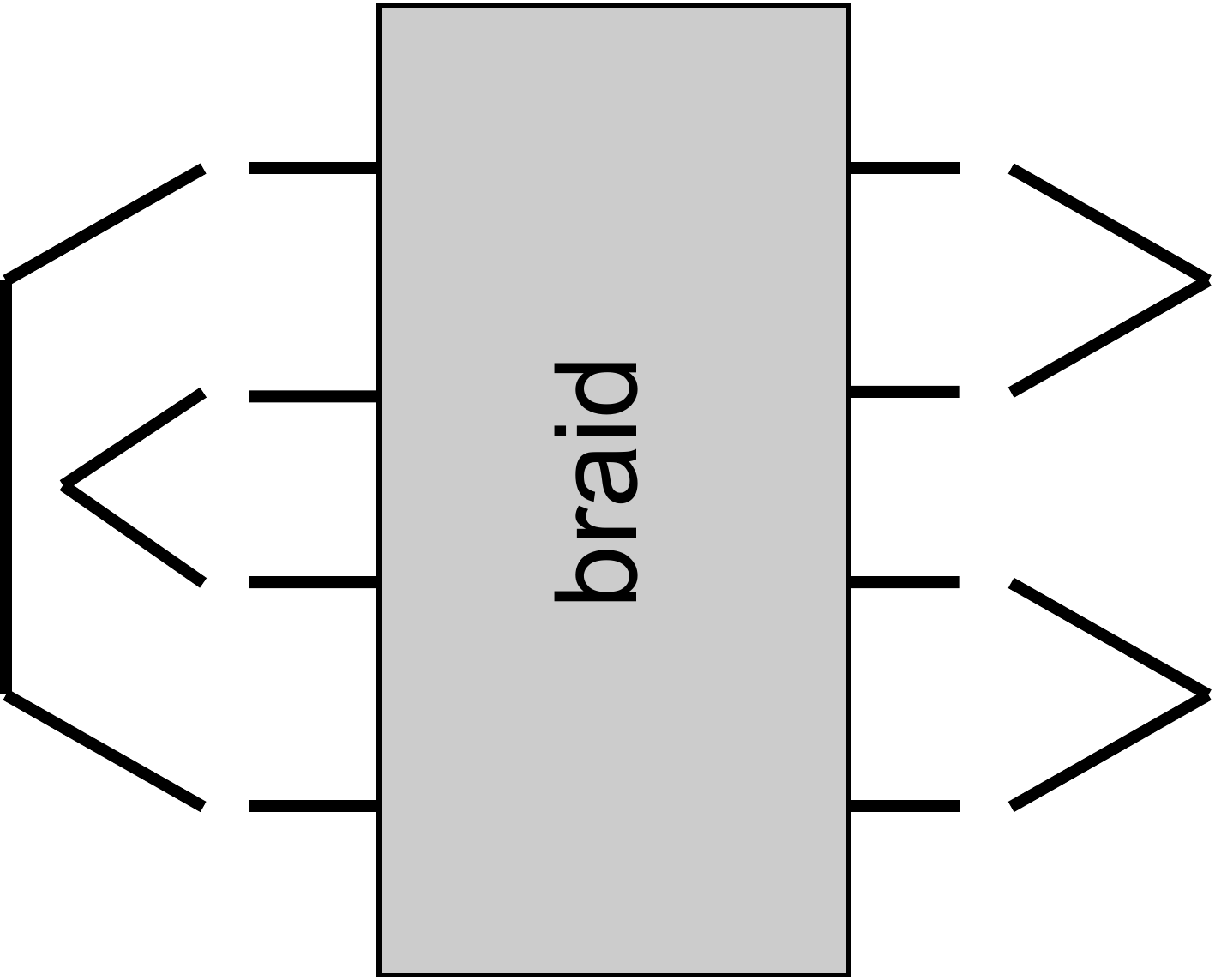}}
    \caption{Braid word presentation of the 2-bridge knot $L_{n,m}$. On the right end, we always close by neighboring pairs (as in \eq{Rstate}). On the left end, the closure prescription depends on what is the leftmost braiding operation in the word: if it is given by $\sigma_2$, which acts on the second and third strands, we must close as in (a) since it is the only non-trivial option; if it is given by $\sigma_1$, which acts on the first pair, the only option is to close as in (b).}
    \label{fig:close}
\end{figure}

In terms of $P\SL$ generators, the $B_3$ word defining the link $L_{n,m}$ is 
\begin{equation}\label{braidwordST}
   U_{n,m}=\sigma_2^{a_1}\sigma_1^{a_2}\sigma_2^{a_3}\cdots = ST^{a_1}ST^{a_2}ST^{a_3}S\cdots\,,
\end{equation}
which has nearly the same presentation as before (c.f. \eq{Unmword}). More precisely, they coincide when $a_1=0$. Note that so far nothing has been said about the particular continued fraction that defines the $a_i$ (hence $L_{m,n}$). An interesting theorem of Schubert \cite{Murasugi,Schubert1956} proves that the isotopy class of unoriented 2-bridge links actually does not depend on the particular choice of continued fraction -- they all give the same link. Therefore, it follows naturally from Proposition \ref{prop1} in the previous section that (for $m>0$) the shortest braid word presentation of a given $L_{n,m}$ is obtained for the regular continued fraction, since it minimizes the $ST$-word in \eq{braidwordST}. In other words, the topological complexity $C_{n,m}$ (i.e., the minimal number of $B_3$ generators needed to represent the link) associated with the unitary operation of \eq{braidcircuit} for an arbitrary 2-bridge link $L_{n,m}$ with $m>0$ is simply
\begin{align}
C_{n,m} = \min_{\{a_i\}}\sum_{i=1}^r|a_i| = \sum_{i=1}^r b_i
\end{align}
where $b_i$ are the Euclidean CF coefficients of $\tfrac{n}{m}$.

It is worth mentioning here that torus knots of the type $(2,2p+1)$ ($p>0$) are included in the rational family. Since $\tfrac{2}{2p+1}=[0;p,2]$, the minimal braid word obtained from the prescription above is simply $U_{2,2p+1}=\sigma_2^p\sigma_1^2=ST^pST^2$. The corresponding complexity is $C_{2,p}=p+2$ (in terms of $B_3$ generators), which is equivalent to $p+4$ generators of $P\SL$. Notice that this is compatible with expression \eq{Cnm} for the complexity of torus knot states studied before, the only subtlety being the $2$ units of framing ($f=2$) that appear naturally here in the language of 2-bridge knots.

\section{Relation with path integral optimization}
\label{sec:pathintegral}

The goal of this section is to connect the knot complexity derived in this paper with the framework of path integral optimization \cite{Caputa:2017urj,Caputa:2017yrh,Bhattacharyya:2018wym,Takayanagi:2018pml}.

In \cite{Caputa:2017yrh}, the authors consider a 2d lattice regularization of a path integral generating a certain target quantum state. The size of the sites may be varied to maximally simplify the numerical computation, and this procedure is shown to be equivalent to minimizing a certain complexity functional associated to the target quantum state. Equivalently, we can keep the lattice coordinate size fixed and act on the metric that, in 2d, can always be brought \emph{locally} to the form $g \xrightarrow[\text{Diff}]{}e^{2\phi}\hat{g}$, for some reference metric $\hat{g}$. This implies that complexity minimization is equivalent to choosing an appropriate Weyl rescaling. We want to show here that large gauge transformations that act on the torus complex structure $\tau$ should also be included in the optimization procedure whenever the space topology is non-trivial. In this way, we provide a more general framework for discussing path integral complexity and establish a clear connection with the specific discussion in the previous sections.

We begin by considering a generic 2d CFT covariantly coupled to a background metric $g$. This metric has the role of defining the lattice size and is not a dynamical field. The path integral we are interested in is the Euclidean time evolution on the plane from $-\infty$ up to a certain time $t_0$, and we impose some boundary conditions for the metric at $t_0$ fixing the lattice size to be $\epsilon$:
\be \label{bc}
ds^2=e^{2\phi(x,t)}(dt^2+dx^2), \hspace{1cm} e^{2\phi(x,t_0)}=\frac{1}{\epsilon^2}\,.
\ee 
The key point here is that the integration measure on a generic field $\varphi$ is \emph{not} invariant under conformal transformations $g\rightarrow e^{2\phi}g$, while it is invariant under diffeomorphisms. This happens because of the requirement to have Gaussian integral normalization:
\be 
\int \mathcal{D}_g \delta \varphi \;\text{Exp}{\int_{\Sigma}d^2 z\sqrt{g}\delta \varphi\cdot \delta\varphi}=1.
\ee 
The integration measure is then  sensitive to a Weyl rescaling such that
\be 
\mathcal{D}_{e^{2\phi}g} \delta \varphi=e^{S_L(\phi)}\mathcal{D}_{g} \delta \varphi
\ee
with $S_L$ the Liouville action \cite{Polyakov:1981rd}
\be \label{liou}
S_L(\phi)=\frac{c}{24\pi}\int_{M}d^2 \sigma (\sqrt{g}\hat{g}^{ab}\partial_a\phi \partial_b\phi +\mu e^{2\phi})\,.
\ee 
Because of this anomaly, the state $\Psi$ prepared by path integration transforms with an overall normalization depending on the choice of $\phi$, and this additional factor is defined as the (exponential of the) complexity of $\Psi$:
\be 
\Psi_{e^{2\phi}g}=e^{S_L(\phi)-S_L(0)}\Psi_{g}\,.
\ee
In the context of AdS/CFT, the metric is interpreted as being induced on the plane by embedding it into $\text{AdS}_3$, with the state $\Psi$ defined at the conformal boundary of AdS space. 

A generalization of this discussion was introduced in \cite{Takayanagi:2018pml}, leading to the following formula
\be\label{pi}
\int D\varphi(x)\  {\rm e}^{-S_{M_{\Sigma}}[\varphi]}\delta\big(\varphi(x,t_0)-\varphi_\Sigma(x)\big) \ = \ {\rm e}^{C_M}\Psi_\Sigma\,.
\ee
Here the constant time surface $\Sigma=\partial M_{\Sigma}$, on which the state $\Psi$ is defined by path integration over $M_{\Sigma}$, does not need to belong to the conformal boundary, but following the surface-state correspondence of \cite{Miyaji:2015yva} can be any convex codimension two surface in (here Euclidean) AdS. The state $\Psi_{\Sigma}$ is argued to be independent of the actual shape of $M_{\Sigma}$ except for the overall normalization. The complexity $C_M$ is then holographically computed as the value of the on-shell gravitational action restricted to the bulk region $N_{\Sigma}$, $I_G(N_{\Sigma})$, where $N_{\Sigma}$ is bounded by $M_{\Sigma}$ and the $t_0$-slice containing $\Sigma$. It is then shown in \cite{Takayanagi:2018pml} that $I_G(N_{\Sigma})=S_L(\phi)$, for $\phi$ being the local cutoff (or, quite equivalently, the Weyl rescaling, see \cite{Takayanagi:2018pml} for details) that brings the metric on $M_{\Sigma}$ to the form (\ref{bc}), showing that the two definitions of complexity agree. 

A further generalization can be obtained by adding a surface $\Sigma_i$ providing initial conditions for the state at $\Sigma$. The path integration along $\tilde{M}_{\Sigma_i\Sigma}$ produces the time evolution $U(\Sigma,\Sigma_i)$ from $\Sigma_i$ to $\Sigma$ that can be used to glue the $M_{\Sigma_i}$ and $M_\Sigma$ surfaces. The total associated complexity can be written as:
\be 
\braket{\Psi_{\Sigma}|U(\Sigma,\Sigma_i)|\Psi_{\Sigma_i}}=e^{C(\Psi_{\Sigma_i})+C(\Psi_{\Sigma})+C(M_{\Sigma_i\Sigma})}\,.
\ee

Let us now move a step forward and consider the case where $M$ does not have a trivial topology, as it was  considered for the disk topology so far. For example, we could pick $\tilde{M}_{\Sigma_i\Sigma}$ to be a cylinder and directly join the two boundaries $\Sigma=\Sigma_i$. By translational invariance along $\Sigma$ and $\Sigma_i$, we can even twist one of the boundaries before gluing and obtain a generic torus $T^2$. When the spacetime topology of $M$ is not trivial, $g \xrightarrow[\text{Diff}]{}e^{2\phi}\hat{g}$ is no longer valid \emph{globally}: the orbit $\mathcal{C}_g$ generated by the diffeomorphism transformation does not cover the whole metric space, which should then be seen as a fiber bundle with fiber $\mathcal{C}_g$ and base the moduli space $\mathcal{M}$. The real dimension of this moduli space is determined by the Riemann-Roch theorem to be, for negative Euler number, $\text{dim}\left(\mathcal{M}\right)=6g-6+3b+2o$, with $g$ the genus, $b$ the number of boundaries and $o$ the number of operator insertions (punctures) on the surface $M$. For zero Euler number, such as a torus or a cylinder, we have two and one moduli respectively (and an equal number of conformal killing vectors)\footnote{There is a small caveat involved in counting the total moduli that comes from the boundary condition (\ref{bc}): this fixes the size of the boundary circle $\Sigma$ so that a modulus less should be included. For instance if we had considered a cylinder, with boundaries $\Sigma_i$ and $\Sigma$ as before, the counting would give zero moduli even if a real modulus would have to be included if we had not imposed (\ref{bc}).
}. In general then $\hat{g}$ depends on the coordinates $\tau$ of $\mathcal{M}$, so that
\be\label{diff}
g \xrightarrow[\text{Diff}]{}e^{2\phi}\hat{g}(\tau).
\ee 
For instance, if we were to integrate over $g$ the integration measure would then split as
\be 
\mathcal{D}g=\text{Jac}\; d\tau \mathcal{D}\phi \mathcal{D}\xi \mathcal{D}\bar{\xi}\,,
\ee
where $\xi,\bar{\xi}$ are local 1-forms parametrizing the diffeomorphism transformation $\delta g_{zz}=\nabla_z \xi_z$ and its complex conjugate. The complexity now should be minimized both along the fiber $\mathcal{C}_g$ and the moduli space base $\mathcal{M}$, so that we should choose the local Weyl transformation \emph{and} the point in moduli space $\tau$. Following the same steps leading to (\ref{liou}), the functional becomes \cite{Ginsparg:1993is}
\be \label{lioutau}
S_L(\phi,\tau)=\frac{c}{24\pi}\int_{M}d^2 \sigma \left(\sqrt{g(\tau)}g(\tau)^{ab}\partial_a\phi \partial_b\phi+2 R(\tau) \phi+\mu e^{2\phi}\right) +\frac{c}{24\pi}\int_{\Sigma}d\sigma \sqrt{h(\tau)} K(\tau) \phi\,,
\ee 
where we have denoted by $R$ the Ricci scalar while $K$ and $h$ are respectively the trace of the extrinsic curvature and the induced metric on $\Sigma$. Equivalently we can invoke holography and obtain the same expression as the gravity on-shell partition function restricted to $N$, which in the case of $M=T^2$ is just the torus interior.

The final piece of information we need is the effect of inserting some primary operator $\mathcal{O}(x_0)$ with scaling dimension $h$ on the surface $M$. Considering that after a Weyl transformation on $\mathcal{O}(x_0)$ this scales with weight $h$, this simply implies that the total complexity should be modified as $S_L-2h \phi(x_0)$.

To finally make connection with the knot complexity, we should ask  the following: what is the effect of a path integral optimization on the Wilson loop that wraps the knot? Note that we have in mind here the realization of 3d gravity as the difference of two $SL(2,\mathbb{R})$ Chern-Simons Lagrangians. The situation here is slightly different from the operator insertion since the Wilson loop is, in our framework, inserted inside the bulk region $N$ and not on its boundary $M$. However, the philosophy is the same: we should vary over the Weyl field $\phi$ and the complex structure $\tau$ in order to optimize the path integration on $M$ and compute the effect on the complexity. As the complexity computed by holography is the saddle point of the gravitational action $-I_G(N)$ restricted to the bulk region $N$, including a bulk Wilson loop just amounts to compute the saddle point for the gravity partition function \emph{with} the Wilson loop insertion; formally, after considering backreaction,
\be \label{prob}
e^{C_M}\sim \int  \mathcal{P}e^{\oint A}e^{-I_G(N)}|_{saddle}.
\ee
The minimization of this quantity is achieved by acting on $M$ by Weyl rescaling, which does not affect the Wilson loop due to the topological nature of the theory, as well as large gauge transformations, that indeed act changing $\tau$. This connects the path integral optimization proposal to our discussion of the complexity of TQFT states. From this argument, it is not clear \emph{a priori} whether we should obtain the same complexity functional for knots that we have derived in this paper. It is an interesting open problem to work out in detail (\ref{prob}) and see what kind of complexity functional it leads to. We leave this for a future work.

\section{Conclusions}
\label{sec:conclusions}

Our main motivation in this paper was to use topological theories to bridge the notions of complexity that have recently been proposed in the quantum gravity and quantum field theory context with more familiar notions in computer science. This is due in part to the finite-dimensional nature of the Hilbert space of such theories in the case of compact gauge groups, which makes the problem tractable, as well as to the topological nature of gravity in $3d$. As a first step, we have focused on $3d$ Chern-Simons theory with compact gauge group and defined the complexity for states associated with torus knots and some connected examples. The key player in our game was the group $P\SL$ and its unitary representations, which naturally generate the quantum evolution of Chern-Simons states in terms of elementary quantum gates $\Sc$ and $\Tc$. 

Since the unitary representations of $P\SL$ are not in general faithful, we distinguished the circuit complexity of quantum knot states $\Cc$ (defined in terms of $\Sc$ and $\Tc$) and the topological complexity $C$ (in terms of $S$ and $T$), which is a characterization of the knots themselves. Topological complexity sets an upper bound for quantum complexity. This bound is saturated in the semiclassical limit of the Chern-Simons theory, in which the correspondence between the knots and the states becomes exact. We have found that in this limit the standard textbook notion of circuit complexity as the minimum number of gates necessary to produce a desired target state from a given reference has a number of interesting properties.

The topological complexity of a state corresponding to a $(n,m)$ torus knot (relative to the unknot) is the length of the shortest word of $S$ and $T$ generators yielding the desired torus diffeomorphism. It can be computed exactly in terms of the particular continued fraction decomposition $[b_1,\ldots,b_r]$ of the rational number $\tfrac{n}{m}$ that has only positive coefficients -- the so-called regular or Euclidean continued fraction. Such presentation is unique up to a trivial ambiguity in the last term and corresponds to the Euclidean algorithm of computing the greatest common divisor of $n$ and $m$. This is the main result proven in Proposition~\ref{prop1}. The result for the topological complexity ${C}_{n,m}$ is given by equation~\eq{Cnm}. 

Interestingly, the topological complexity computed for the first few thousand torus knots grows at a much slower rate when compared with their crossing number. This is consistent with the idea that torus knots are intrinsically simpler than other knots with the same number of crossings. Also, for all our checked examples, \eq{Cnm} calculated using the Euclidean continued fraction gives the same result as the one computed with the ancestral path continued fraction, introduced in \cite{beardon2012} as an explicit example of continued fraction having the minimal possible number of coefficients. The latter defines a geodesic path on the Farey graph connecting $\infty=\tfrac{1}{0}$ and $\tfrac{n}{m}$, but \emph{a priori} has no reason to minimize the $\Sc\Tc$-word. Hence, the complexity \eq{Cnm} might play a deeper role from the point of view of number theory and the modular group as an index characterizing geodesic continued fraction presentations. The conjectured relation was formulated as Proposition~\ref{prop2}. Finally, we numerically observed that the complexity grows roughly as $C_{n,m}\sim \tfrac{n}{m}$ for large $n$.

We also discussed interesting geometric interpretations of the topological complexity formula in the  Farey tesselation of the hyperbolic plane, which could be relevant for understanding the connection of the present discussion with quantum gravity and its tensor network models. The sum of Euclidean continued fraction coefficients $b_i$ is itself also a \q{geodesic distance} in the upper-half plane between the origin and the fraction $\tfrac{n}{m}$ in the sense of \cite{Kulkarni}. Equivalently, it corresponds to the depth of $\tfrac{n}{m}$ on the Stern-Brocot tree of positive rationals or to the number of ideal triangles traversed in the dual Farey graph. The number $r$ of terms in the continued fraction corresponds to the distance to $\infty$ travelling along a special path of Farey neighbors, which is by itself composed of a series of geodesic arcs. Due to the fact that $r$ is bounded from above for fixed $m$, the main contribution to the complexity for large $n$ is due to the number of triangles. This seems to be well-approximated by the area below the geodesic curves in Figure~\ref{fig:farey}, suggesting a connection of our results with the holographic subregion complexity \cite{Alishahiha2015,Abt2018}. We notice, however, that our discussion is more in the spirit of ``holography in the moduli space'' \cite{Manin2002} rather than the usual spacetime holography.

In Section~\ref{sec:pathintegral}, we embedded our discussion in the wider framework of path integral optimization, recently introduced in \cite{Caputa:2017urj}, by showing how the original argument should be modified for non-trivial topological spaces: it is necessary to also optimize the moduli space of the spacetime. We formally derived the complexity functional in this language and an interesting future problem would be to work out its explicit form and better understand the relationship with our results for the complexity of knot states.

Finally, we also discussed circuit complexity in terms of braid group generators (i.e., using the presentation of knots as closures of braid words) in the case of 2-bridge knots in the Hilbert space of Chern-Simons on the sphere with four punctures. These are also characterized by two coprime numbers and have close connections with continued fractions, which again allows to easily obtain the topological complexity as a shortest word on $B_3$ generators via Proposition \ref{prop1}. As in case of torus knots, the topological complexity is exact but only coincides with the quantum state complexity in the semiclassical limit of Chern-Simons, otherwise setting an upper bound on it.

It is interesting to know whether a general formula for quantum complexity exists beyond the semiclassical limit, but a case by case study is needed here since we are not aware of any general statement about modular representations with arbitrary gauge group $G$ and level $k$. This is related to understanding the structure of unitary representations of mapping class groups, see for example, a recent work~\cite{Singerman:2018}. Possible generalizations of quantum complexity formulas, as well as the case of more general classes of knots shall be addressed in a forthcoming work \cite{ongoing}.

\subsection*{Acknowledgements}
The authors thank F.~Ares, T.~Fleury, M.~Lencsés and A.~Mehra for useful discussions. GC is also grateful to Ian Short for insightful correspondence on continued fractions. GC, FN and AP would like to thank the financial support from the Brazilian ministries MEC and MCTIC. The work of DM was supported by the Russian Science Foundation grant No 18-71-10073.

\appendix{}

\section{Torus knot states}
\label{sec:appendix}

In this section we review some facts about the torus knot states and explain the difference between the type states considered, for example, in~\cite{Labastida:1995kf} and the complement states of reference~\cite{Balasubramanian:2016sro}.

As mentioned in section~\ref{sec:CS}, TQFT assigns a torus some Hilbert space $\Hc(T^2)$. In the case of CS TQFT we parameterize it additionally with the group and level: $\Hc(T^2;G,k)$. States in this Hilbert space are three-dimensional spaces $M$ attached to $T^2$, such that torus is their boundary $\partial M=T^2$. Canonical choice of basis on $\Hc(T^2;G,k)$ corresponds to considering solid tori with Wilson lines parallel to the longitude of the torus (unknots) and coloring them with integrable representations of $G$. Hence, any state in the Hilbert space can be expanded over the basis of vectors $|i\rangle$, with $i=1,\ldots,N$ labeling integrable representations,
\be
|\Psi\rangle \ = \ \sum\limits_{i=1}^N\psi_i\,|i\rangle\,.
\ee
We note that the basis $|i\rangle$ is orthonormal, which follows from the scalar product
\be
\langle i|j\rangle \ =\  Z(S^2\times S^1;\bar{i},j) \ = \  \delta_{ij}\,,
\ee
which is nothing but the topological invariant (Chern-Simons partition function) of two unlinked circles in $S^2\times S^1$ parallel to $S^1$ and colored by representations $j$ and the conjugate of representation $i$. Such an invariant is unity if $i=j$, and zero otherwise.

One set of states in $\Hc(T^2)$ to consider are those represented by solid tori containing torus knots inside $|j_{m,n}\rangle$. Expanding over the basis means gluing such a torus with a solid torus with an unknot Wilson line inside. The second torus must be inverted, to represent the dual space $\Hc^\ast(T^2)$, and the representation of the unknot should be conjugated. The expansion coefficients are then
\be
\psi_i(j_{m,n}) \ = \ \langle i|j_{m,n}\rangle = Z(S^2\times S^1;\bar{i}_{1,0},j_{m,n})\,, 
\ee
which are the $S^2\times S^1$ invariants of a torus knot in representation $j$ and an unknot colored with $\bar{i}$, both parallel to $S^1$ and unlinked. In particular, the coefficient with trivial $i=\emptyset$ gives an invariant of the torus knot in $S^2\times S^1$.

Recall that if one wants to compute the invariants of torus knots in $S^3$, one should take the solid torus with a knot ${(m,n)}$ and glue it with a complement of this solid in $S^3$, which is some vector $\langle \Omega|\in \Hc^\ast(T^2)$. The complement is also an inverted torus, but with an opposite identification of the contractible and non-contractible cycles. Hence,
\be
\label{Omega}
\langle\Omega| \ = \ \sum\limits_i \langle i|\Sc_{0i}\,,
\ee
where $\Sc_{ij}$ is a representation of the modular matrix $S$ on $\Hc(T^2)$. Therefore,
\be
Z(S^3;j_{m,n}) \ = \ \langle\Omega|j_{m,n}\rangle \ = \ \sum\limits_i \Sc_{0i} Z(S^2\times S^1;\bar{i}_{1,0},j_{m,n})\,.
\ee

Equation~(\ref{jnm}) in section~\ref{sec:modular} tells that a torus knot $(m,n)$ inside a solid torus can be obtained from an unknot by a modular transformation. In terms of the fundamental cycles $\alpha$ and $\beta$ on the torus, the unknot corresponds to $\omega=\beta$, while an arbitrary torus knot is $\omega=n\alpha+ m\beta$. These two linear combinations of fundamental cycles can be connected by the modular group element
\be
U^{m,n}\ = \ \left(\begin{array}{cc}
 m & p \\
 n & q 
\end{array}
\right)
\ee
Note that this can also be viewed as a M\"obius transformation
\be
z\ \to \ \frac{mz+p}{nz+q}\,, 
\ee
of the \q{rational number} $1/0=\infty$ into the number $m/n$. If one has a unitary representation of $\SL$ acting on $\Hc(T^2)$, then any torus knot state can be obtained from the unknot by an appropriate $\SL$ transformation:
\be
|j_{m,n}\rangle \ = \  \Uc^{m,n}|j_{1,0}\rangle \ = \ \sum\limits_k  \Uc^{m,n}_{kj}|k\rangle\,.
\ee
Since the representation is unitary, states $|j_{m,n}\rangle$ also form a basis. Torus knot invariants can be expressed as matrix elements of $\SL$ elements,
\be
\label{S3pols}
Z(S^3;j_{m,n}) \ = \ \langle\Omega|\,\Uc^{m,n}|j_{1,0}\rangle \ = \ \sum\limits_i \langle i|\Sc_{0i}\sum\limits_k  \Uc^{m,n}_{kj}|k\rangle \ = \  \sum\limits_k  \Sc_{0k}\Uc^{m,n}_{kj}
\ee

One might also be interested in states which have coefficients~(\ref{S3pols}) as the amplitudes defining a knot state, that is
\be
|\Kc_{m,n}\rangle \ = \ \sum\limits_j Z(S^3;j_{m,n}) |j\rangle\,.
\ee
To represent such a state in terms of spaces, it appears more convenient to think of their conjugate versions
\be
\langle\Kc_{m,n}| \ = \ \sum\limits_i \psi_i \langle i|\,.
\ee
Note that $\langle\Omega|$ in equation~(\ref{Omega}) is an example of such a dual state. Its coefficients are invariants $\psi_i=Z(S^3,R_i)$ of an unknot in $S^3$. Now consider a similar state of the form
\be
\sum\limits_i\langle \Kc_{m,n}(i)|\,\Sc_{0i} \ = \ \sum\limits_{i,k}\langle k|\,\Sc_{0i}(\Uc^{m,n})^\dagger_{ik}\,,
\ee
where we expand over a different basis in the dual space, but with the same coefficients as in~(\ref{Omega}). Sandwiching these dual states with basis elements $|j\rangle$ gives the same invariants in $S^3$, albeit complex conjugated. 

What is the space interpretation of states $\langle\Kc_{m,n}|$? They correspond to spaces with a torus boundary, which produce $S^3$ with a torus knot inside, when glued with a solid torus with a Wilson line $j_{1,0}$. Hence such $\langle\Kc_{m,n}|$ are complements of a tubular neighborhood of torus knot $(m,n)$ in $S^3$.

\bibliographystyle{JHEP}
\bibliography{refs}

\end{document}